\documentclass[10pt,conference,compsoc]{IEEEtran}

\usepackage{booktabs} 
\usepackage{balance} 
\usepackage{amssymb}
\usepackage{mathtools}
\usepackage{color}
\usepackage{verbatim}
\usepackage{amsmath}
\usepackage{xspace}
\usepackage{graphicx}
\usepackage[caption=false]{subfig}
\usepackage{xr}
\usepackage{multirow}
\pdfoutput=1

\def\skipnoindent{\vskip0.1in\noindent}

\newcommand{\Paragraph}[1]{\smallskip\noindent{\it \textbf{#1:}}}
\newcommand\floor[1]{\lfloor#1\rfloor}

\usepackage{amsthm}

\newtheorem{thm}{Theorem}

\newtheorem{definition}{Definition}
\usepackage{booktabs}
\usepackage{tabularx}
\usepackage{adjustbox}

\ifCLASSOPTIONcompsoc
  \usepackage[nocompress]{cite}
\else
  \usepackage{cite}
\fi

\newcommand{\acronym}{\emph{NoiSense}\xspace}

\begin{document}

\title{\acronym: Detecting Data Integrity Attacks on Sensor Measurements using Hardware based Fingerprints}

\author{Chuadhry Mujeeb Ahmed,
        Aditya Mathur,
        and~Martin Ochoa, \\ 
        Address: Singapore University of Technology and Design \\ 
E-mail: chuadhry@mymail.sutd.edu.sg}

\IEEEtitleabstractindextext{%
\begin{abstract}


In recent years fingerprinting of various physical and logical devices has been proposed with the goal of uniquely 
identifying users or devices of mainstream IT systems such as PCs, Laptops and smart phones. On the other hand, the 
application of such techniques in Cyber-Physical Systems (CPS) is less explored due to various reasons, such as 
difficulty of direct access to critical systems and the cost involved in faithfully reproducing realistic scenarios. 
In this work we evaluate the feasibility of using fingerprinting techniques in the context of realistic 
Industrial Control Systems related to water treatment and distribution. 
Based on experiments conducted with  44 sensors of six different types,  it is shown that noise patterns due to microscopic imperfections in hardware manufacturing can be used to uniquely identify 
sensors in a CPS with up to $97\%$ accuracy.  The proposed technique can be used in
 to detect physical attacks, such as the replacement of legitimate sensors by faulty or manipulated sensors. 
We also  show that, unexpectedly, sensor fingerprinting can effectively detect advanced physical attacks such 
as analog sensor spoofing due to variations in received energy at the transducer of an active sensor. Also, it can be leveraged to construct 
a novel challenge-response protocol that exposes cyber-attacks.

\end{abstract}

}

\maketitle



\IEEEdisplaynontitleabstractindextext


%

\IEEEpeerreviewmaketitle

\section{Introduction}
A Cyber Physical System (CPS) is a distributed computing system that is sensing and actuating on the physical world~\cite{alur2015principles,nist_cps_def}. For instance large critical infrastructures, such as power generation and water treatment plants, are particular examples of CPS, also known as Industrial Control Systems (ICS) \cite{EAlee2008}. An ICS consists of cyber components such as  Programmable Logic Controllers (PLCs), sensors, actuators, Supervisory Control and Data Acquisition (SCADA) workstation, and Human Machine Interface (HMI) elements interconnected via a communications network.  The PLCs and the SCADA workstation operate in concert to control the physical process. The widespread use of such technology, and the increasing number of incidents involving it~\cite{ukranian_case2016analysis,slay_miller_2008,stuxnet}, has raised concern on the security of CPS~\cite{cardenas2009challenges}. 

Different from an attack against conventional IT systems, an attack on a CPS may directly result in physical damage to property (for instance due to an explosion~\cite{aurora_attack,German_steelmill_attack}) and the loss of human life (such as operators of the CPS or people depending on a critical infrastructure). Security requirements of a CPS are thus different from those found in conventional cyber security~\cite{cardenas2009challenges}, in particular the \emph{physical} integrity of the system and its availability are often more important than confidentiality aspects~\cite{Gollmann2016}. 
Moreover, in a CPS an attacker may not only compromise the computing elements but he might partially also do so for physical components or the physical environment. This is illustrated for instance by the recent attack of~\cite{drone_Son2015}, where a crash is induced in a drone by means of a sound signal that confuses the gyroscope, or by recent analog sensor spoofing attacks~\cite{shoukry2015,yasser-2013,sampling_race2016,ghosttalk_2013}. This makes CPS security challenging, since we need to expand traditional attacker models to include physical and cyber-physical characteristics of a system~\cite{marco_cpdy2016}, and consequently there is a need for novel security solutions in the intersection of these worlds.

In a CPS thus, the \emph{veracity}~\cite{Gollmann2016} of sensor data is paramount to the security of the system~\cite{surveyCPSsec_IEEE2016_sensorAtt_Mclaughlin}. The impact of various kinds of integrity attacks on sensor values has been studied mathematically in the control theory community, including false data injection~\cite{Mo2012_falsedatainjection}, replay attacks~\cite{Mo2009_replayAttacks}, and stealthy attacks~\cite{dan2010stealth}.  To implement such attacks in practice, attackers can spoof sensor values either by physical means, as discussed above, for instance by analog spoofing attacks~\cite{shoukry2015,yasser-2013,ghosttalk_2013}, or by tampering with the communication channel between a sensor and a controller by means of a classical man-in-the-middle attack~\cite{urbina_CCS2016limiting}. 



Digital or cyber-attacks can be performed by attackers who have compromised the CPS communication network and can spoof digitized sensor readings. 
For instance, reverse engineering of the well-known Stuxnet worm \cite{stuxnet,bruno_mo_physicalwatermarking_2015} showed that it would attempt to 
perform a replay attack of normal states of a CPS while carrying out an attack to conceal it from human operators. 
Physical attacks can be carried out by attackers who have physical access to a CPS, such as malicious insiders~\cite{quarta_sp2017,slay_miller_2008,McLaughlin2010,CPSsecsurvey_JIoT2017_insiders_phyA}, or attackers that have access to 
a CPS distributed over a large geographical area~\cite{CPSsecsurvey_JIoT2017_insiders_phyA}. For instance, water distribution networks or smart grids are usually spread over hundreds of miles across  city, and 
thus it is hard to guard all the components~\cite{raheem2016,sridhar_smartgrid2012}. In \cite{McLaughlin2010} ease of physically tampering energy meters at the consumer 
end, without leaving any evidence, is shown. Authors in~\cite{drone_Son2015,ghosttalk_2013} report  sensor spoofing attacks by considering physical proximity of a 
few centimeters to the victim. Therefore, physical attacks are also a realistic threat model for a CPS~\cite{raheem2016,quarta_sp2017,sridhar_smartgrid2012,shoukry2015,sampling_race2016,ghosttalk_2013}. 

 {\acronym}: In this work we propose a non-intrusive sensor fingerprinting method to authenticate sensors transmitting measurements to one or more PLCs. Device fingerprinting ideas based on clock skews, modulation schemes and transmission circuitry, have been reported  in the literature~\cite{kohno2005,boris2009,faria2006,remley2005}. However, sensors in an ICS are not computationally powerful enough to exhibit the above mentioned fingerprints~\cite{raheem2016}. Thus, we seek an answer to the question, \textit{Do sensors in a real world ICS have unique fingerprints?} It is known that hardware imperfections during the manufacturing process exhibit some unique physical behaviors that are useful for profiling and fingerprinting~\cite{dey-2014}. In particular, we observe that \emph{noise} (imperfections in measurements), an otherwise undesirable feature of sensors, strongly depends on such manufacturing imperfections. These variations affect each device differently and thus are hard to control or reproduce~\cite{gerdes2006}, making it challenging for an attacker to imitate sensor noise pattern. 

A technique, referred to as \acronym, is designed to fingerprint sensors found in ICS. \acronym creates a fingerprint for a sensor based on a set of time domain and frequency domain features that are extracted from the sensor noise. A machine learning algorithm is used to distinguish an individual sensor from others. In particular, a  multi-class Support Vector Machine (SVM) is used  to identify each sensor from a dataset, comprising of a multitude of industrial sensors. Experiments were  performed on a total of $44$ sensors including $23$ low cost ultrasonic level sensors and $21$ sensors of different types in an  operational water treatment and distribution facility accessible for research\,\cite{swat2016,wadi2017}. Sensor identification accuracy is observed to be as high as $97\%$, with a low of $90\%$. It is also  shown that the proposed scheme is scalable for tens of sensors and that the sensor fingerprint is stable over time. True positive rate for sensor identification is observed to be  $100\%$ for most of the sensors and false positive rate as low as $0\%$.  The cost of actual industrial scale sensors is typically a few thousand dollars, which led us to design experiments for low cost sensors. However, experiments performed on various type of industrial sensors make it a representative study for a general ICS framework. 

 \acronym  
has certain advantages that make it suitable for  deployment in an ICS: (a) It is \emph{non-Intrusive}, as no modifications in CPS hardware are required. (b)~It is a passive fingerprinting technique that identifies a sensor in an operational process without affecting its intended functionality.
(c)~It is a low cost solution that can be used at the design stage and also in an operational  CPS without any significant additional  cost. (d)~It does not require any functional modifications to the system or control logic, other than the addition of  specific code in a controller, such as in a PLC, to detect sensor  tampering. One of the strongest features of the proposed scheme is that, it is able to detect attacks originating from physical (analog) as well as cyber (digital) domains.


The major contributions of this work are thus:
\begin{itemize}
\item A novel sensor fingerprinting framework  that is based on sensor noise, and is a function of hardware characteristics of a  device.
\item A detailed evaluation of the proposed fingerprinting method, for a class of invasive and non-invasive physical attacks.
\item Extensive empirical performance evaluation on realistic  testbeds as well as using controlled lab experiments.
\item A novel \emph{challenge-response} protocol based on sensor noise fingerprinting.
\end{itemize}
This work evaluates \acronym in the context of water treatment and water distribution testbeds~\cite{swat2016,wadi2017}. Commonly found industrial sensors are studied, but without loss of generality, these analysis are applicable   to other industrial applications.


\section{Preliminaries}
\label{background}

\subsection{Threat Model}\label{threat_model_sec}
In this work, we consider specific physical and cyber attacks on sensor measurements in an ICS. First, we lay down our assumptions about the attacker, followed by justification for such assumptions. Attacker's goals, objectives and attack scenarios are also explained in detail.

\Paragraph{Assumptions on Attacker}\label{assumptions_on_Attacker}
We consider a  strong adversary who is able to launch cyber and/or physical attacks. In an ICS sensors, actuators and PLCs communicate with each other via communication networks. An attacker can compromise these communication links in a classic \emph{Man-in-The-Middle (MiTM)} attack~\cite{urbina_CCS2016limiting,SridharAdepu_AsiaCCS2016_L1Attacks,SAmin_2013_StealthyAtt_canal}, for example, by breaking into the link between sensors and PLCs. Recent studies, demonstrate malware attacks on PLCs~\cite{Anand_ESORICS2017_PLC_ladderlogicbomb,saman_PLCMalware_NDSS2017}. Besides false data injection in sensor readings via cyber domain, an adversary can also physically tamper a sensor, to drive a CPS into an unstable state. Therefore, we need to authenticate sensor measurements, which are transmitted to a controller. 
A \emph{malicious insider}, is an attacker with physical access to the plant and thus to its devices such as sensors. An attacker can physically replace or tamper sensors. However, such an attacker does not necessarily, need to be an \emph{insider}, because critical infrastructures, e.g.,for  water and power, are generally distributed across large areas~\cite{raheem2016,sridhar_smartgrid2012}.  An  \emph{outsider}, e.g., end user, can also carry out a physical attack on sensors such as  smart energy monitors. Physical attacks, invasive and non-invasive, have also been considered as a threat model in traditional IT systems~\cite{physical_attacks_smartcard2005,Anderson1996}.


\Paragraph{Attacker's Goals$(\mathcal{G})$}
An attacker may choose his  goals from a set of intentions\,\cite{sridhar2015} such as  performance degradation, disturbing a system property, and damaging a component.

\skipnoindent \textit{Physical Damage} $(\mathcal{G}1)$: An attacker can damage devices in a plant including pumping stations and other electrical appliances. Doing so may cause injury to individuals in the plant or at a large scale by altering the devices for instance by flooding the surroundings of the plant with wastewater as in the Maroochy Shire incident~\cite{slay_miller_2008}. 

\skipnoindent \textit{Reduction in Quality/Quantity of Product} $(\mathcal{G}2)$: An attacker can inject faults and defects in the product for a general industrial control system. In particular, for water networks, it can under-dose or over-dose certain chemicals to compromise the water quality. It can also reduce the production, resulting in water supply outages for the consumers.   

\skipnoindent \textit{Utility Theft} $(\mathcal{G}3)$: Utility theft is usually done by the end consumer. It can be achieved by tampering different kinds of monitors/sensors~\cite{McLaughlin2010}. One example is water theft from canals/irrigation systems~\cite{amin2013a,amin2013b}. 

\Paragraph{Attacker's Intermediary Objectives$(\mathcal{O})$}
\skipnoindent \textit{Concealment} $(\mathcal{O}1)$: An attacker wants to conceal the attack as much as possible~\cite{shoukry2015}. For instance, in certain  ICS~\cite{swat2016}, it has been observed that disconnecting the sensor from the controller or cutting the sensor wire to the controller, does not raise any alarm by default, making such an attack undetectable by default. In case of a stealthy cyber attack, an attacker manipulates the sensor measurements according to a precise mathematical model~\cite{SAmin_2013_StealthyAtt_canal,dan2010stealth,Carlos_Justin_CDC2016_stealthyAtt} to conceal it's actions. 

\skipnoindent \textit{Inaccurate Measurements and Signal Masking} $(\mathcal{O}2)$: Physically replaced malicious sensors should report an inaccurate and imprecise data to disturb the product quality. An attacker can achieve the same objective of inaccurate measurements by false data injection. Also, an adversary may prevent the sensor to measure the true physical quantity. This can be achieved by data injection in \emph{SCADA} workstation~\cite{urbina_CCS2016limiting,SridharAdepu_AsiaCCS2016_L1Attacks} and also by masking the true quantity by overpowering the original signal~\cite{shoukry2015,sampling_race2016}. 

\skipnoindent \textit{Deceiving Controller} $(\mathcal{O}3)$: The attacker attempts to deceive the controller into believing that the data received is from a legitimate sensor. \\

\Paragraph{Attacker Characterization$(\mathcal{C})$}

\skipnoindent \textit{Attacker Profile} $(\mathcal{C}1)$:
Most likely attacker profile requiring physical access is that of a \emph{malicious insider}~\cite{Rocchetto2016}. However, for attacker goal \emph{$\mathcal{G}$3} (utility theft), a consumer may be a malicious entity. Another strong candidate for such an attacker is     \emph{casual outsider} (malicious contractor or technician) as defined in~\cite{quarta_sp2017}. Considering the critical nature of these utility infrastructures a nation state or a terrorist attacker profile could not be ignored~\cite{Rocchetto2016}. An \emph{outsider} can break into the system using cyber domain  and may not necessarily need physical access to the plant. 

\skipnoindent \textit{Technical Capabilities} $(\mathcal{C}2)$:
We assume an adversary, (\emph{insider or outsider}), has the complete knowledge of the workings of the plant  and sensing devices in particular. An end user could employ services of a contractor, a criminal entity, for sensor tampering, to achieve the utility theft goal(~\emph{$\mathcal{G}$3}). \\



\Paragraph{Attack Scenarios $(\mathcal{S})$}\label{attack_scenarios}
We  categorize attack scenarios into two categories, namely physical domain and cyber domain.\\

\noindent\underline{Physical (Analog Domain) Attacks:}

 \skipnoindent \textit{Sensor Replacement Attack} $(\mathcal{S}1)$: An attacker replaces one or more sensors in the plant by new, perhaps malicious,  sensors.   This way the attacker can  manipulate the sensor readings to drive the system to an undesirable state. The attacker can inject arbitrary sensor  data that may not be detected by cyber attack detection schemes based on statistical methods such as Cumulative Sum (CUSUM)\,\cite{SAmin_2013_StealthyAtt_canal}.

\skipnoindent \textit{Sensor Swap Attack} $(\mathcal{S}2)$: Sensor swap attack uses the legitimate control logic to achieve an attacker's goal. Rather than modifying the control logic, an attacker feeds a controller with measurements from another process. The idea is to use the existing control logic to drive the plant to an insecure state. 
In appendix~\ref{swap_attack_example}, an example swap attack and it's consequences on a water treatment testbed are analyzed.

 \skipnoindent \textit{Sensor Saturation Attack}$(\mathcal{S}3)$: Sensor saturation attack is similar to jamming attacks and is executed by injecting power to the sensor's receiver~\cite{sensor_saturationAttack_infusionpump_usenix2016}. A constant value is measured by a sensor under attack and it is not able to monitor the real physical quantity. Another way to achieve a similar attack in industrial sensors is to block the physical medium (either by wireless channel jamming or cutting the wired media) between the sensor and remote \emph{I\O} unit. This attack also returns a constant reading to the \emph{SCADA} workstation. An attacker can remain stealthy for such attacks as system raises no alarm. However, it is obvious that transmitted constant reading does not contain sensor noise component which enables \acronym, to detect such attacks or faults in components. 
 
 
 \begin{figure}[!htb]
\centering
\includegraphics[scale=.35]{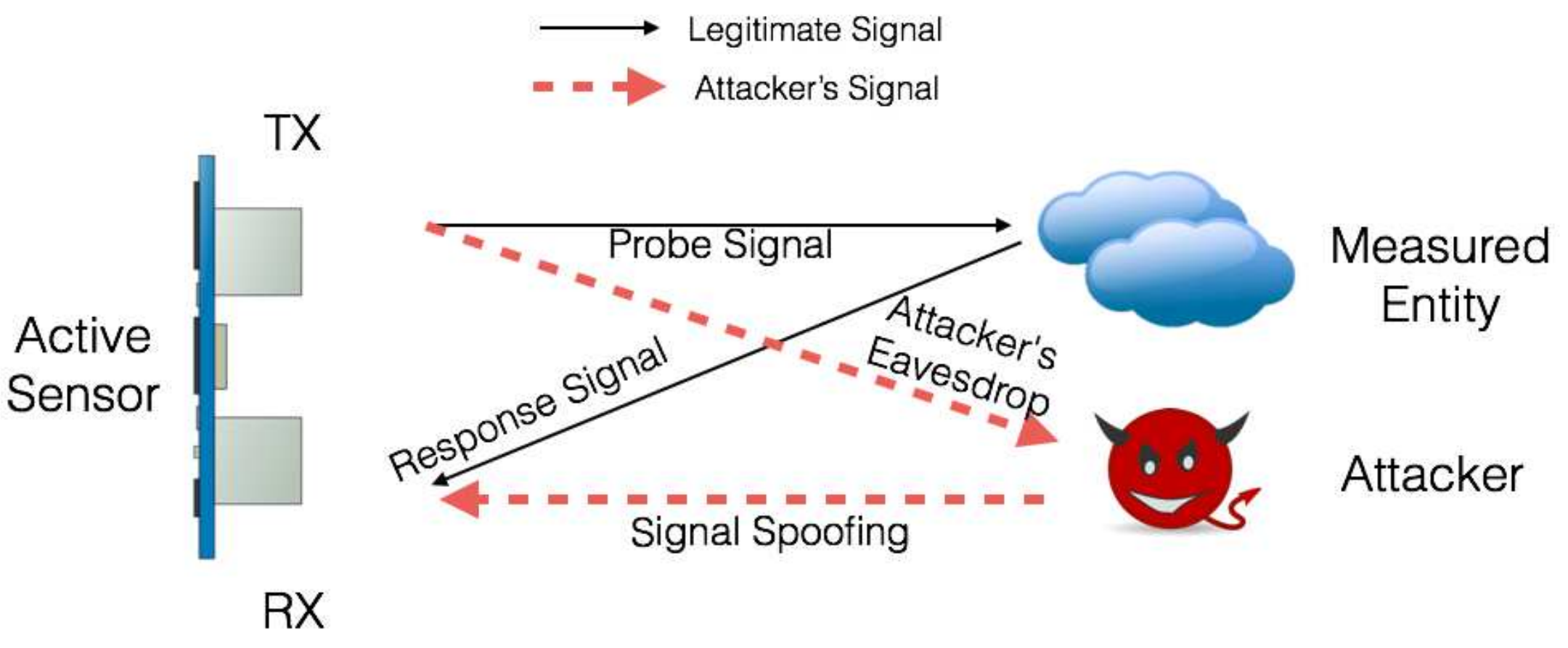}
\caption{Analog Sensor Spoofing Attack.}
\label{analog_sensor_spoofing}
\end{figure} 

\skipnoindent \textit{Analog Sensor Spoofing Attack}$(\mathcal{S}4)$: The attack scenarios ($\mathcal{S}1$,$\mathcal{S}2$,$\mathcal{S}3$), mentioned above,  need physical access to the plant and are also invasive requiring sensor tampering. Another type of physical attack can be non-invasive where an attacker spoofs sensor measurements by disturbing the surrounding environment. This can be achieved by bringing another malicious entity~(same type of device as the victim) in proximity with the victim device. One such attack is analog spoofing of the measured quantity~\cite{shoukry2015} for a class of active sensors. Active sensors transmit a probe signal which interacts with the quantity to be measured and a response signal is received and analyzed to measure the physical quantity. But if an attacker has physical access to the sensing environment, it can change the response signal before it reaches the receiver. An example of such an attacker is shown figure~\ref{analog_sensor_spoofing}. Transmitter (TX) of the active sensor transmits a probe signal and waits for the response signal to be received. Rather than getting the legitimate response signal, an attacker transmits a fake signal. When this fake signal is received at the receiver (RX) of the active sensor (victim), it would not be able to differentiate between a legitimate and the signal generated by an attacker.


\skipnoindent\underline{Network-based (Cyber Domain) Attacks:}

\skipnoindent \textit{Digital Domain Sensor Swap Attack}$(\mathcal{S}5)$: This attack is similar in concept to $\mathcal{S}2$ but an attacker does not necessarily need  physical access to sensors. In control logic of the plant, it is possible to exchange the tags for the sensors for their respective PLCs~\cite{Anand_ESORICS2017_PLC_ladderlogicbomb}. $PLC_1$ can be made to read sensor~2, and $PLC_2$ can be made to read sensor~1. The effects of such a swap would be same as in the case of physical sensor swap. However, such a change will not be reflected at the  \emph{SCADA} workstation where human operators are monitoring the process 

\skipnoindent \textit{False Data Injection in Sensor Measurements}$(\mathcal{S}6)$: 
This attack can be executed as \emph{MiTM} whereby an attacker modifies the unencrypted sensor data transmitted to a PLC~\cite{SridharAdepu_AsiaCCS2016_L1Attacks,urbina_CCS2016limiting}. Such a modification will change the noise fingerprint of a sensor and  enable \acronym to detect the attack.  

\skipnoindent \textit{Stealthy Attacks}$(\mathcal{S}7)$: 
An attacker modifies a sensor measurement and attempts to hide it's presence. In the literature~\cite{Ahmed_AsiaCCS2017_stealthyAtt,Rizwan_ESORICS2017_stealthyAtt,CPSweek2016_stealthy_replayATT,Carlos_Justin_CDC2016_stealthyAtt,SAmin_2013_StealthyAtt_canal} specific stealthy attacks have been designed to stay undetected for a range of statistical detectors. To achieve the \emph{stealthiness} an attacker must modify sensor measurement in a way so that it can maximize the damage and remain undetected. To achieve this objective an attacker will inevitably change the sensor noise which will enable \acronym, to detect the presence of the attacker.

\skipnoindent \textit{Replay and Advanced Sensor Spoofing Attacks}$(\mathcal{S}8)$: 
In a replay attack an attacker records the system states for the normal operation of the system. Then it replays recorded states during an attack to hide itself from operators and digital domain intrusion detection systems~\cite{CPSweek2016_stealthy_replayATT,Mo2009_replayAttacks}. An example of such an attack is \emph{Stuxnet}~\cite{stuxnet}. A replay attack on sensor measurements might not be detected by \acronym because sensor noise will also be replayed. 
Another example is of a very powerful cyber attacker with an ability to learn noise fingerprint of a sensor. It can modify sensor measurements to arbitrary values and also add noise patterns for that sensor, making it strong enough to remain undetected by \acronym. We extend the idea of \acronym for the case of such a powerful cyber attacker by proposing a novel \emph{challenge-response} scheme  presented in section~\ref{challenge_response_protocol_sec}.


\begin{figure}[!tb]
\centering
\includegraphics[scale=0.3]{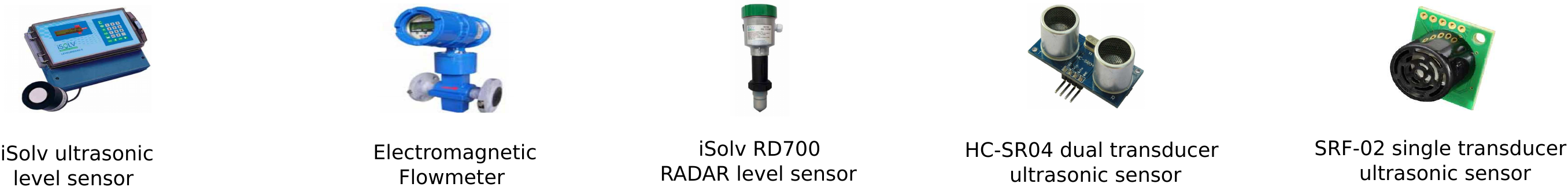}
\caption{Sensors used in experiments.}
\label{sensors_used}
\end{figure} 

\subsection{Sensing Technologies}\label{sensing_technologies_sec}
In this section we explain the basic working principle of the sensing technologies under study. This insight in sensor construction and functionality is an aid in understanding the sources of sensor noise and fingerprints. 

\Paragraph{Ultrasonic Level Sensors}\label{ultrasonic-technology}
Water treatment testbeds use ultrasonic sensors based on a piezoelectric (PZT ceramic) material transducer. The level of water in a tank is calculated by measuring the return time of the acoustic wave after hitting the water surface.  Several factors  contribute to variations in the measurements obtained from  ultrasonic sensors. These  measurements depend on the speed of sound which changes according to the surrounding temperature~\cite{jenny-2013}. Besides temperature, obstacles like tank walls reflect echo sooner than expected, contributing towards noise in the measurements. The acoustic impedance of the PZT transducers also depends on temperature thus adding another source of noise \cite{coutard-2005}. Thermal and polarization noise are the main sources of voltage fluctuation in piezoelectric ceramics~\cite{petr-2011}.    

\Paragraph{Microwave Level Sensors}
The microwave level/distance sensor, often called RADAR (Radio Distance and Ranging), works in a similar way as ultrasonic sensors. A microwave pulse is emitted by the antenna that travels at the speed of light and upon hitting the surface of the target it is reflected back and received at the same antenna. These antennae are designed to have a 50$\Omega$ resistance so that once connected with a cable of characteristic impedance of 50$\Omega$, maximum power transfer takes place from the antenna. The sensor under consideration is designed to operate at 26\,GHz with a beam angle of $22^o$ and 1$\mu$W effective radiated power\,\cite{flotech_radar}. However, in practice these specifications have deviation for the same type and design of an antenna due to manufacturing imperfections and installation inaccuracies. For example, antenna connection with a cable will result in impedance variations\,\cite{accuracy_radar}. Also, beam angle and radiation pattern  varies for each antenna leading to deviations from theoretical design resulting in different range resolution that is ultimately reflected in sensor noise \cite{antenna_pattern}.

\Paragraph{Electromagnetic Flow Meters}
The electromagnetic flow meters follow Faraday's law of induction according to which a voltage is induced by an electrically conductive fluid passing through a magnetic field. In an electromagnetic flow meter, the medium acts as the electrical conductor when flowing through the flow meter tube, and the induced voltage is proportional to the average flow velocity (the faster the flow rate, the higher the voltage). A commercial electromagnetic flow meter is shown in Figure~\ref{sensors_used} \cite{flotech_flowmeter}. 
It's internal structure consists of a pair of coils mounted on the top and bottom of an electrically insulated flow tube. A pair of electrodes protrude through the flow
tube wall perpendicular to the pipe axes and largely normal to the direction of the generated magnetic field. Noise in these sensor readings come from the area of the electrodes and size of the electro-magnets generating electromagnetic field $B$. The installation and alignment of electrodes and coils will result in different stray capacitance and noise \cite{flowmeter2006}.

\section{Design of \acronym}
\label{sensorprint_design_sec}
Figure\,\ref{sensor_print} shows the  steps involved in composing a sensor fingerprint. The proposed scheme begins with data collection  and then divides data into smaller chunks to extract a set of time domain and frequency domain features. Features are combined and labeled with a sensor ID. A machine learning algorithm  is used for sensor classification.

\subsection{Data Collection}
Data is collected for different types of industrial sensors listed in Table\,\ref{sensors_used_table}. We collect data for the level sensors when the process is static, i.e. tank levels are constant.  For  flow meters, data is collected when process is dynamic, i.e. water is flowing through the pipes and hence a non-zero flow rate is observed. The objective of data collection step is to extract sensor noise. For the case of level sensors, when the process is running, an error in sensor reading is a combination of sensor noise and process noise (water sloshing etc.). Extracting process noise from a dynamic process is a challenging task, given  that the noise parameters vary with change in the process state. Therefore, a set of experiments for level sensors, are  designed to obtain sensor measurements, when a process is  not active. Tanks  store water to provide to subsequent stages for processing. However,  these processes are not always active as water demand is not constant. If there is no inflow and outflow of water, in a tank, the corresponding level sensor measurement should be a constant.  Nevertheless, there are  fluctuations in the sensor measurements, as a result of sensor noise or temperature variations. In the controlled lab environment, temperature is controlled, and it affects all the sensors in a similar way. 
 
We have a limited number of level sensors in water treatment testbed, but we diversified experiments to validate the proposed idea. A total of $3$~level sensors are installed on top of $3$~water tanks. Each sensor is placed on all the three water tanks, to collect data for all possible sensor-tank combinations. The data is analyzed, in time and frequency domains, to examine the noise patterns, which are found to follow  Gaussian  distribution. Sensors are profiled using variance and other statistical  features in the noise vector. The experiment is run, to obtain  sensor profile, so that it can be used for later  testing. A machine learning algorithm is used to profile sensors from fresh readings (test-data). Design and testing, of \acronym, is feasible in settings of a water treatment plant. According to the control logic, when a tank is filled up to a specified limit, a pumping station turns OFF and the water level stays at a constant value. Since these water treatment systems, have multiple processing stages, there are instances when there is neither an incoming flow nor an outgoing flow from a tank, i.e. we have a constant water level in the tank. During these instances, we can record the data and match it with previously generated fingerprint. Fingerprints can be generated over time or at the commissioning phase of the plant. 

For  flow meters, if there is no water flow, we receive a value of $0$ from the sensor. In the testbeds under study, the flow of water between different stages is controlled by a pump. Therefore, water flow through a pipe should be constant but flow rate data transmitted by the electromagnetic sensor is noisy. We treat these fluctuations in the sensor data as noise and collect this information during the times when process is running. For the case of low cost ultrasonic sensors, we designed an experiment where a constant water level is to be measured. Noise in each sensor is extracted and data  analyzed to generate the fingerprint. Thus, by combining the level sensor and flow meter features, we are able to monitor the process, in both static (constant water level) and dynamic (water flowing) states. 

 
 \begin{figure}[!htb]
\centering
\includegraphics[scale=.35]{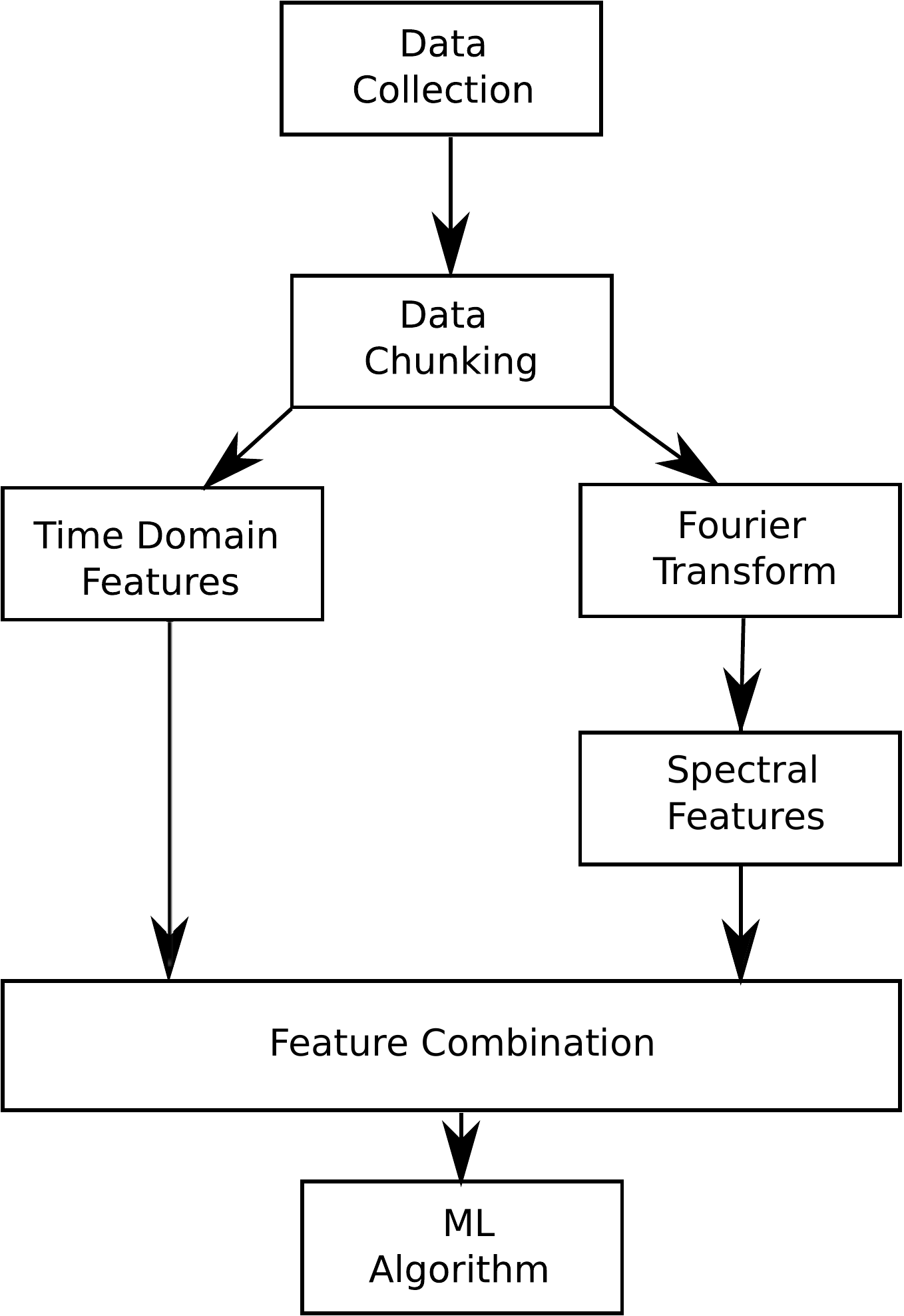}
\caption{Components of  \acronym, design.}
\label{sensor_print}
\end{figure}  

\begin{table}
\begin{center}
\caption{List of features used.  }
\label{features}
 \begin{tabular}[!htb]{|l | l|} 
 \hline
 {\bf Feature} & {\bf Description}  \\ 
 \hline
 Mean & $\bar{x} = \frac{1}{N} \sum_{i=1}^{N} x_i$  \\ 
 \hline
 Std-Dev & $\sigma = \sqrt[]{\frac{1}{N-1}\sum_{i=1}^{N}(x_i - \bar{x}_i)^2}$  \\
 \hline
 Mean Avg. Dev & $D_{\bar{x}} = \frac{1}{N}\sum_{i=1}^{N} |x_i - \bar{x}|$  \\
 \hline
 Skewness & $\gamma = \frac{1}{N}\sum_{i=1}^{N}(\frac{x_i - \bar{x}}{\sigma})^3$  \\
 \hline
 Kurtosis & $\beta = \frac{1}{N} \sum_{i=1}^{N}(\frac{x_i - \bar{x}}{\sigma})^4 - 3$  \\
 \hline
 Spec. Std-Dev & $\sigma_s = \sqrt[]{\frac{\sum_{i=1}^{N}(y_f(i)^2)*y_m(i)}{\sum_{i=1}^{N}y_m(i)}}$  \\
 \hline
 Spec. Centroid & $C_s = \frac{\sum_{i=1}^{N}(y_f(i))*y_m(i)}{\sum_{i=1}^{N}y_m(i)}$  \\
 \hline
 DC Component & $y_m(0)$  \\ [1ex] 
 \hline
 \multicolumn{2}{p{0.45\textwidth}}{Vector $x$ is time domain data from the sensor for $N$ elements in the data chunk. Vector $y$ is the frequency domain feature of sensor data. $y_f$ is the vector of bin frequencies and $y_m$ is the magnitude of the frequency coefficients.} 
\end{tabular}
\end{center}
\end{table}

\subsection{Feature Extraction}
Data is collected from sensors at a sampling rate of one every second. Since data is collected over time, we can use raw data to extract time domain features. We used the Fast Fourier Transform (FFT) algorithm\,\cite{welch1967} to convert data to frequency domain and  extract the spectral features. In total, as in  Table\,\ref{features}, eight features are used to construct the fingerprint. 

\noindent \emph{Data Chunking}: After data collection from the sensors, the next step is to create chunks of dataset. An important question to answer is:  \emph{How many chunks of data are needed to train a well-performing machine learning model?} Also, we want to know, \emph{how much time it will take to perform a test?} so that a decision about the attacks can be made within a specified time. To this end we create data chunks of varying size for each sensor. After dividing a sensor's data (total of $N$ readings) into $m$ chunks, each chunk of $\floor{\frac{N}{m}}$,  the feature set $<F(C_i)>$ for each data chunk $i$ is calculated For each sensor, there are $m$ sets of features $<F(C_i)>_{i \in [1,m]}$. For $n$ sensors one can use $n \times m$ sets of features to train the multi-class SVM. Supervised learning was used  for sensor identification which has two phases-- training and testing. For both the phases, chunks are created in a similar way as explained above. 

\noindent \emph{Size of Training and Testing Dataset}: An important question to address is:  \emph{How many feature sets are needed for training the classifier and how many for testing?} For a total of $m$ feature sets for each sensor, at first  half ($\frac{m}{2}$) were used for training and half ($\frac{m}{2}$) for testing. To analyze the accuracy of the classifier for smaller feature sets was reduced during the training phase, number of feature sets starting with $\frac{m}{2}$. 
Classification is then carried out for the following corresponding range of feature sets for $\mbox{Training}: \{\frac{m}{2}, \frac{m}{3}, \frac{m}{4}, \frac{m}{5}, \frac{m}{10} \}$, and for $\mbox{Testing}: \{ \frac{m}{2}, \frac{2m}{3}, \frac{3m}{4}, \frac{4m}{5}, \frac{9m}{10} \}$, respectively. In section~\ref{evaluation_sec}, empirical results are presented for such feature set and the one with the best performance is chosen for further analysis of the proposed scheme. For the classifier a multi-class SVM is used~\cite{libsvm}, as briefly described in Appendix~\ref{sec:svm}.

\subsection{Performance Metrics}
In Table~\ref{sensors_used_table},  sensors of same type are grouped together. Each sensor is assigned a unique ID and a multi-class classification is applied to identify each sensor among those of the same type and model group. To evaluate the performance,  identification accuracy is used as a performance metric. Let $c$ be the total number of classes.  $TP_i$ as true positive for class $c_i$ when it is rightly classified based on the ground truth. False negative $FN_i$ is defined as the wrongly rejected, and false positive $FP_i$ as wrongly accepted, device. True negative $TN_i$ is the rightly rejected class. The overall accuracy ($acc$) for total of $c$ classes is defined as in the following. 

\begin{equation}
acc = \frac{\sum_{i=1}^{c} TP_i+\sum_{i=1}^{c}TN_i}{\sum_{i=1}^{c}TP_i+\sum_{i=1}^{c}TN_i+\sum_{i=1}^{c}FP_i+\sum_{i=1}^{c}FN_i}.
\end{equation}

\noindent The True Positive Rate (TPR) and False Positive Rate (FPR) are  defined as follows,

\begin{equation}\label{tpr}
\mbox{True Positive Rate (TPR)} = \frac{TP}{TP + FN},
\end{equation}
\begin{equation}\label{fpr}
\mbox{False Positive Rate (FPR)} = \frac{FP}{FP + TN}.
\end{equation}

\begin{table}
\begin{center}
\caption{List of sensors in our study.}
\label{sensors_used_table}
 \noindent\begin{tabularx}{\linewidth}{|X | c | c|} 
 \hline
 Type & No. of Devices & Class \\ 
 \hline
 Ultrasonic level sensor & 3 & C1  \\ 
 \hline
 RADAR level sensor & 2 & C2  \\
 \hline
 HCSR-04 dual transducer 5V & 15 & C3  \\
 \hline
 HCSR-04 dual transducer 3.3V & 5 & C4 \\ 
 \hline
  Electromagnetic Flow Meter & 8 & C5  \\
 \hline
 SRF02 single transducer Sonar & 3 & C6 \\ 
 \hline
 Differential Pressure Transmitter & 8 & C7 \\
 \hline

\end{tabularx}
\end{center}
\end{table}

 \begin{figure}[!htb]
\centering
\includegraphics[scale=.35]{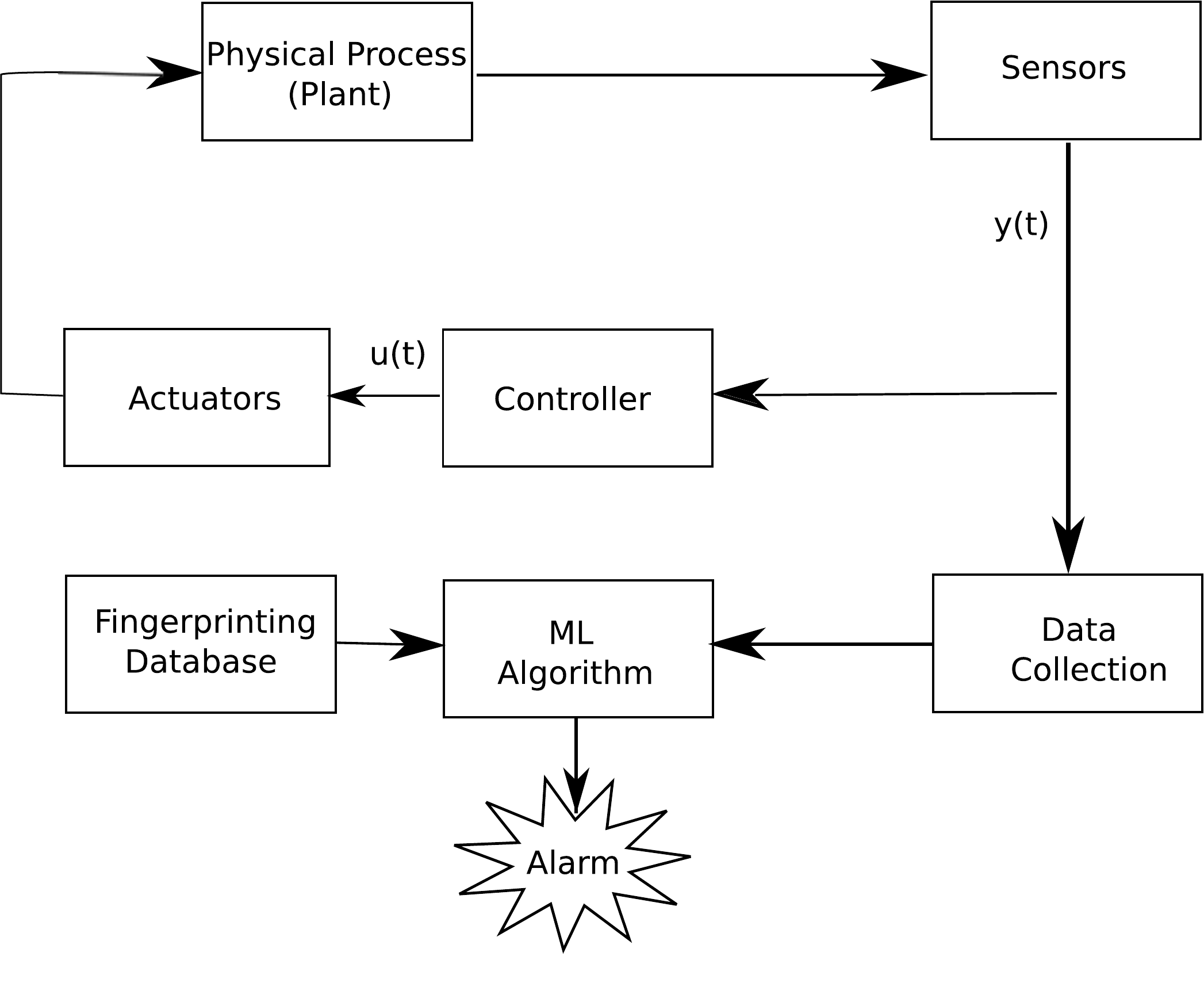}
\caption{System Model for Proposed Methodology.}
\label{system-model}
\end{figure} 

\section{Evaluation}\label{evaluation_sec}
The  underlying idea  for the proposed method is to find a fingerprint from the extracted noise of sensors under consideration. Figure~\ref{system-model} shows the integration of the proposed scheme to any existing ICS without need of additional hardware installation. Also, the proposed method is passive and does not disrupt the functionality of the control system. 


\noindent\textbf{Research Questions:}
The following  research questions are the focus of the remainder of this work.

\begin{itemize}

\item [\bf RQ1:] \emph{Does a unique fingerprint exist for each sensor?}

\item [\bf RQ2:] \emph{How much data is needed to identify a sensor?}

\item [\bf RQ3:] \emph{How accurately can the sensors be identified?}

\item [\bf RQ4:] \emph{Is a fingerprint  unique with respect to multiple  sensors such as is the case in a large CPS?}

\item [\bf RQ5:] \emph{Is the sensor fingerprint stable between different runs of the experiment over a period of time?}

\item [\bf RQ6:] \emph{Is the sensor fingerprint based method able to detect analog sensor spoofing attacks?}

\end{itemize}

\noindent Sensor data is analyzed for different sets of sensors to answer the above questions. Extensive data collection  and analysis is carried out to obtain  the fingerprint and show that indeed this is a valid fingerprint. The following sections describe    the experimentation setup as well as results obtained by applying the proposed attack detection scheme.  

\subsection{Experimentation Setup}
Experiments described in this article were conducted on an operational   water distribution and water treatment systems. One  testbed is a fully operational research facility and a  scaled down water treatment plant. The other testbed is a water distribution network\,\cite{swat2016,wadi2017}. Additional information on these testbeds is in  the appendices~\ref{swat-testbed} and~\ref{wadi-testbed}. 
 Besides using the sensors in the two testbeds,  23 small ultrasonic sensors were also used to show that the proposed scheme can scale well to a larger set  of sensors. Table~\ref{sensors_used_table} lists  the  sensors used in this study. 
These sensors are representative of those found in a range of industrial plants as fluid storage and flow is a common process~\cite{bela-2003}.

 
 \begin{figure*}[!htb]
\centering
\includegraphics[width=18cm, height=4.5cm]{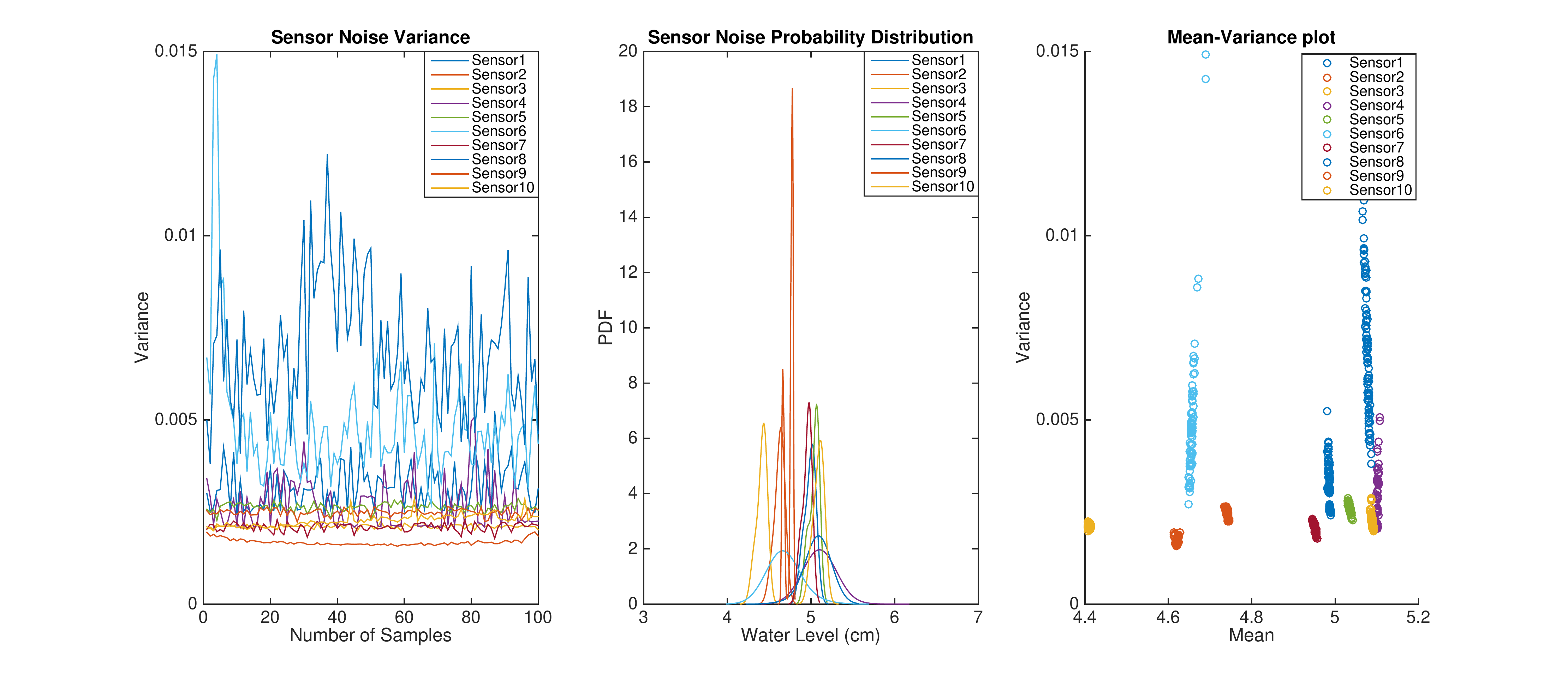}
\vspace{-2em}
\caption{Noise data from ten ultrasonic sensors of same type (HC~SR04). Left: Variance for each sensor noise vector is shown for different data chunks of each sensor. Middle: Distribution of each sensor noise vector is shown for a data chunk. Right: Including mean of the noise improves the separation for individual sensors.}
\label{variance_plot}
\end{figure*} 

\Paragraph{Ultrasonic Level Sensors (Water Treatment Testbed)}
Experiments were performed on a portion of the six stage water treatment testbed~\cite{swat2016}. Three ultrasonic level sensors, available in this testbed,  were used in the experiments.  Data was collected first from three level sensors in their original tank installation and also by mounting on the other tanks. These experiments were  performed for several hours each day and over several months. An SVM model was trained over the data obtained in original positions of the sensors. For each sensor, testing data was obtained by placing the sensor over the other tanks. This experiment demonstrated: a)~changing the process (tank) does not change the  fingerprint, and b)~fingerprint is stable and valid for long term, e.g., it is valid for data collected at different times over a period of six months.


 
 \begin{figure}[!htb]
\centering
\includegraphics[scale=.35,keepaspectratio]{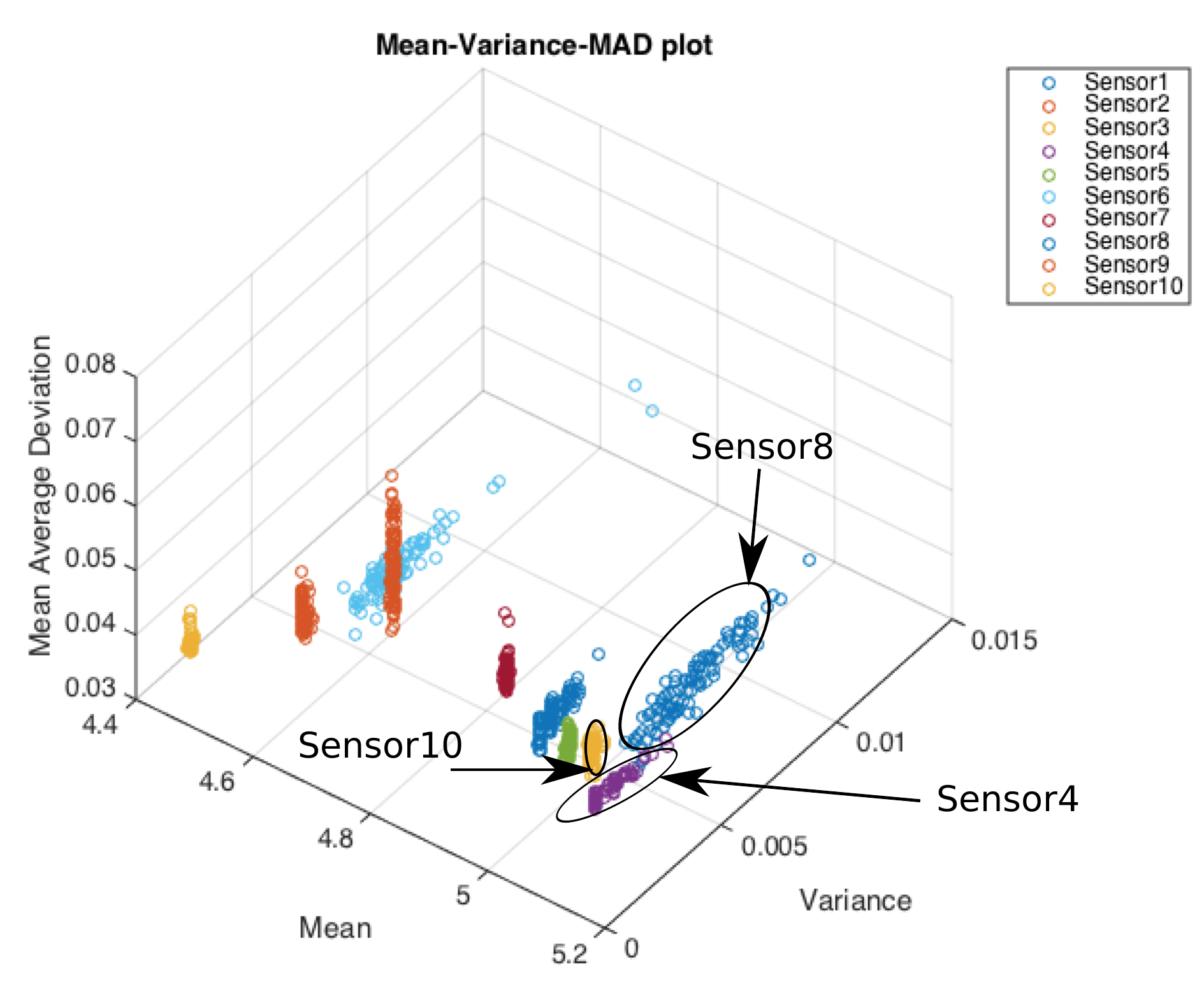}

\caption{Three features have improved the clustering and unique identification of the sensors as compared to  using only the variance or the mean. This is a 3D representation of the right-most plot in Figure~\ref{variance_plot}.}
\label{mad}
\end{figure} 

\Paragraph{Flow Meters (Water Treatment Testbed)}
Electromagnetic flow meters are installed in pipes to monitor the water flow. Data was collected while the plant was running. For flow meters, sensor swap is not as simple as level sensors, since these are installed in-line on a water pipe. Therefore, a 5-fold cross validation results are reported in Table~\ref{all_sensors_accuracy} for the four flow meters (from each testbed) used in the study.   Several measurement vectors were obtained for the water filling process for up to six days continuous run of the plant. The noise vector for each process was extracted and feature set  obtained.

\Paragraph{RADAR Level Sensors (Water Distribution Testbed)}
Two radio frequency based level sensors are available for experiments in the water distribution testbed\,\cite{wadi2017}. Experiments were  performed for three days and data collected. A 5-fold cross validation was performed using SVM and sensors identified with high accuracy.

\Paragraph{Dual Transducer Ultrasonic Sensors (HCSR04)}
 A limited number of industrial sensors were available to perform the experiments. Details of such sensors are given above. It is costly to acquire a lot of such sensors. Therefore, a set of experiments was designed to test our approach on $20$ low cost ultrasonic sensors. The working principle of these low cost sensors is the same as for those used in industrial scale sensors. The goal of these experiments was to explore \emph{whether  a  large set of sensors can  be uniquely fingerprinted and identified  with high accuracy}? Among these $20$ sensors,  $15$ are double transducer ultrasonic sensors which require $5V$~DC to operate, while the remaining $5$ sensors require $3.3V$ DC to operate.  This variety of operation voltages adds to the diversity in the experiments.  

The setup to conduct the experiments on ultrasonic sensors is shown in Figure~\ref{hcsr04}(Appendix~\ref{supporting_figures_appendix}). It is composed of the following main components: 1)~a tank is created with a small half-filled water glass, 2)~breadboard to hold the ultrasonic sensor, 3)~Arduino board (controller) that is connected to sensor, and 4)~a server to collect and store data from sensors over WiFi. The Arduino board has a microcontroller and several input/output pins that serve as an  interface for external circuitry\,\cite{arduino}. 
Ultrasonic sensors can be easily  mounted and removed from the breadboard. Same water level is ensured in the tank for all the sensors. Multiple rounds of experiments were performed to ensure that noise is free of the physical sensor arrangement. Sensor readings were collected over three hours with a sampling time of 1~second. Several chunks of data were extracted from these sets of readings and features  extracted for each chunk. After collecting  data for three hours we removed the sensor from the tank and put it back and collected data for an additional two  hours. This is to demonstrate that the fingerprint can be recovered even after disturbing the sensor alignment. 

\Paragraph{Single Transducer Sonar Sensors (SRF02)}
Three  small single transducer sonar range finders were also included in the experiments. The basic working principle of these sensors  is the same as that of dual transducer ultrasonic sensors though with minor differences. These sensors (SRF02) use a single transducer for both transmission and reception, and the minimum range is higher than the dual transducer sensors used in the experiments. The minimum measurement range varies from around 17-18\,cm. For this reason  a distance measuring experiment was performed with these three sensors rather than measuring water level in small glass tanks. These experiments were controlled through the  Arduino board and placed at the same place to measure distance between sensors and the ceiling. Experiments were run for $6$ hours with one run of $3$ hours and a second run of $3$ hours to explore the stability of the sensor fingerprint.

\subsection{Existence of Fingerprint}
\noindent \textbf{RQ1}: \emph{Does a unique fingerprint exist for each sensor?} A  limited number of sensors were available in the water utility testbeds. Hence,    additional  low cost ultrasonic sensors are included to explore the existence of fingerprints for many sensors of the same type and model. To demonstrate  the existence of fingerprint,  ten dual transducer ultrasonic sensors (HCSR04) from the same manufacturer were used. All  ten sensors were mounted on the same water tank. Data was collected for 3~hours and  many chunks of the collected data  taken for analysis. Each chunk consists of 300 readings from the sensor. Figure~\ref{variance_plot} shows results for the collected data. The plot on the left  shows the variance of noise vector from each sensor for all  chunks. It is  observed that some of these sensors have a unique noise variance and can be distinguished from each other but there remain few sensors that have similar noise pattern in terms of noise variance. The  middle pane is a plot of the distribution of the noise vector from each sensor. It also shows that sensors can be distinguished based on noise statistics. However, there remain  overlaps among some  sensors. Right pane shows 2-D clustering of the sensors. Sensors can be distinguished more precisely by using one more feature of sensor's noise i.e. mean value. The scatter plot on the right hand side clusters each chunk with its respective mean and variance. The separation is quite clear but  there remain  overlaps, e.g.,  sensor4, sensor8 and sensor10. We need additional features to further eliminate such overlaps. In Figure\,\ref{mad}, by adding one more feature, i.e.  mean average deviation, sensor4, sensor8 and sensor10 can be distinguished. In the following sections we show that by using additional  features it is possible to achieve high accuracy for sensor identification.  Details of feature set used are in Table\,\ref{features}.


\begin{table*}[ht!]
\centering
\caption{Identification evaluation for each sensor. \textbf{FPR: }False Positive Rate, \textbf{TPR:} True Positive Rate.}
\begin{adjustbox}{max width=\textwidth}
{\footnotesize
 \begin{tabular}{|c c c c c c c c c c c c c c c c c c c c c|} 
 \hline
  & S1 & S2 & S3 & S4 & S5 & S6 & S7 & S8 & S9 & S10 & S11 & S12 & S13 & S14 & S15 & S16 & S17 & S18 & S19 & S20\\ 
 \hline
 TPR & 1  &  0.87  &  1  &  1  &  1  &  0.97 &   1  &  1  &  0.99 & 0.98  &  1  &  0.98  &  0.99 &   0.98  &  0.86  &  0.97  &  1  &  0.94  & 0.97 &   1\\ 
 \hline
 FPR & $3.4e^{-3}$   & $0.9e^{-3}$   &      0   &      0  &       0   &      0    &     0     &    0    &     0  & $6.2e^{-3}$ &        0 &   $1.2e^{-3}$ &   $2.2e^{-3}$  &  $1.2e^{-3}$   & $7.4e^{-3}$  &  $2.2e^{-3}$   &      0  &  $0.6e^{-3}$ & 0    &     0 \\ [1ex] 
 \hline 
\end{tabular} }
\end{adjustbox}

\label{detection_mean}
\end{table*}

\begin{table*}
\begin{center}
\caption{Classification accuracy for 20 small ultrasonic sensors. 
$m$ is the total number of readings recorded from each each sensor. $\frac{m}{3},\frac{2m}{3}$ means one third of the data is used for training and two third is used for testing respectively. 
The first column in the table shows chunk size of data used to extract features.}
\label{20sensors_accuracy}
 \begin{tabular}[!htb]{|c | c | c | c | c | c|} 
 \hline
 Chunk Size (s) / Data Segmentation & $\frac{m}{2},\frac{m}{2}$ & \textbf{$\frac{m}{3},\frac{2m}{3}$} & $\frac{m}{4},\frac{3m}{4}$ & $\frac{m}{5},\frac{4m}{5}$ & $\frac{m}{10},\frac{9m}{10}$   \\ 
 \hline
 60 & 95.13\% & 95.56\% & 95.14\% & 93.56\% & 88.52\%  \\ 
 \hline
 \textbf{120} & 96.43\% & \textbf{96.43\%} & 95.6\% & 93.86\% & 92.61\%   \\
 \hline
 250 & 97\% & 95.82\% & 94.74\% & 94.37\% & 91.62\% \\
 \hline
 500 & 96.72\% & 95.59\% & 95.20\% & 94.13\% & 85.09\%  \\
 \hline
 700 & 96.74\% & 95.66\% & 91.51\% & 90.55\% & 77.08\% \\
 \hline
 1000 & 96.88\% & 93.25\% & 90.62\% & 83.75\% & 69.25\%  \\ [1ex] 
 \hline

\end{tabular}
\end{center}
\end{table*}

\begin{table*}
\begin{center}
\caption{Overall Result}
\label{all_sensors_accuracy}
 \begin{tabular}[!htb]{|c | c | c | c|} 
 \hline
 Sensor  &  Type and Model & Number of Sensors & Identification Accuracy  \\ 
 \hline
 Ultrasonic Level Sensor & iSOLV LevelWizard II & 3 & 90\% \\ 
 \hline
 Electromagnetic Flowmeter (SWaT) & iSOLV EFS803/CFT183 & 4 & 96\%  \\
 \hline
 Dual Transducer Ultrasonic Level Sensor & HC-SR04 (5V) & 15 & 97.65\% \\
 \hline
 Dual Transducer Ultrasonic Level Sensor version 2 & HC-SR04 (3~5V) & 5 & 97.36\% \\ 
 \hline
 Sonar: Single Transducer Range Finder & SRF02 & 3 & 90\%\\
 \hline
  Electromagnetic Flowmeter (WaDI) & iSOLV EFS803/CFT183 & 4 & 98.2\%  \\
 \hline
 Differential Pressure Transmitter & iSOLV SPT 200 & 8 & 92.5\% \\
 \hline
 RADAR: RF Level Sensor & iSOLV RD700 & 2 & 99\%\\ [1ex] 
 \hline

\end{tabular}
\end{center}
\end{table*}

\subsection{Sensor Identification Accuracy}
\noindent \textbf{RQ2}: \emph{How much data do we need to identify a sensor}? We start our analysis with $20$ dual transducer small ultrasonic sensors. The goal now is to find  how much training data is adequate to identify sensors with high accuracy. Also, we are interested in knowing the chunk size of data to extract features. Table~\ref{20sensors_accuracy} shows results from our analysis. The first row shows the size of training and  testing data set, respectively, where $m$ is total number of chunks for a sensor. The first column in the table shows chunk size starting from $60$~seconds to $1000$~seconds. It is  observed that too small a chunk size does the accuracy is slightly lower as compared to chunks in the middle range. A large chunk size,  as for example $1000$, limits the chunks and hence leads to a small a feature set. Doing so results in lower identification accuracy. If we move towards right in the row, the size of the training data set is reducing, which also results in lower accuracy. This result is intuitive. However, we select a chunk size of $120$~seconds which is not too small that it could not capture sensor noise statistics and not too large that  we have to wait too long before reaching a decision about authenticity of the sensor. Also, it was  decided to divide the entire sensor data set  in three parts and used one third for training for further experiments. For $20$ low cost ultrasonic sensors, this choice of chunk size and training data set provided  $96.43\%$ accuracy for sensor identification as shown in  Table~\ref{20sensors_accuracy}.

For a chunk size of $120$~seconds, and one third of data set for training,  the SVM classifier,  a multi-class SVM classification for $20$ dual transducer ultrasonic sensors was carried out. It was possible to distinguish   sensors with an accuracy of $96.43\%$. Figure~\ref{confusion_matrix} (Appendix D) shows the same result visually with a plot of the confusion matrix as a heat map for these $20$ sensors. The horizontal axis represents actual sensor ID, while the vertical axis represents the predicted sensor ID by the SVM classifier. We note  that nearly all  sensors are accurately identified, and hence, on diagonal of the confusion matrix, the  prediction accuracy is close to $100\%$. 

Table~\ref{detection_mean} shows the TPR and FPR  for each of the sensors. Eq.\,\ref{tpr}, gives us the percentage of rightly classified sensor (TPR), while Eq.\,\ref{fpr} gives the percentage of mis-classification of sensor (FPR). Ideally,  $FPR = 0$ and $TPR = 1$ for  perfect classification. Table~\ref{detection_mean} shows the results for $20$ ultrasonic sensors labeled as $S1, S2,...,S20$, where the TPR is $100\%$ and FPR is $0\%$ for nearly all sensors.  

\skipnoindent  \textbf{RQ3}: \emph{How accurately can we identify sensors}? Table~\ref{all_sensors_accuracy} shows  the sensor identification accuracy for all  $44$ sensors used in our study. This table lists the type and model of sensors, number of sensors, and identification accuracy obtained in the experiments.  It is noted that   the lowest identification accuracy is $90\%$ for $5$ sensors and, the remaining  $27$  sensors have identification accuracy above $97\%$ for their sensor type, respectively. These results highlight the significance of  \acronym.

\subsection{Scalability of the \acronym}
\noindent \textbf{RQ4}: {\em Is a fingerprint  unique against many  sensors such as is the case  in large plants?} Data from the  experiments reported here does not   provide a clear answer to this question  as we had limited sensors in the lab to test the idea of fingerprinting. We also do not have  any theoretical proof of  scalability. However, we experimented on tens of small ultrasonic sensors, which gives an indication of scalability of \acronym. The experiment was designed using a  set of $20$ dual transducer ultrasonic sensors.  $5$ sensors were selected randomly  out of these $20$ sensors and  identification accuracy calculated using SVM. Then,  $10$ sensors were selected randomly, followed by $15$ and ultimately all $20$ sensors were selected. The idea is to explore if identification accuracy drops as  the number of sensors to be identified is increased. The results are shown in Table~\ref{5to20sensors_accuracy}. Identification accuracy does not drop by increasing the number of sensors in the identification pool. This result is intuitive as the idea of \acronym results from the assumption that each device has unique noise characteristics due to manufacturing imperfections. By increasing the number of devices, the confusion for a classifier does not necessarily increase as each device brings it's own unique fingerprint. This shows that the idea of \acronym is scalable for a larger number of sensors.  

\subsection{Stability of  \acronym}
\noindent \textbf{RQ5}: {\em Is the sensor fingerprint stable between different runs of the experiment over a period of time? } Several experiments were performed to answer this question.  Data  collection phase for the level sensors installed in the water treatment testbed was spread over a period of six months. Results shown in Table~\ref{all_sensors_accuracy}, for the case of \textit{iSOLV Level Wizard\,2(level sensor in SWaT)} are representative of this experiment. From these results we  conclude that sensor fingerprints are stable over long periods of time (i.e. six months in our study). 

Besides collecting data over long periods of times on the  water treatment testbed, another experiment was designed  to test with low cost ultrasonic sensors. To see if the sensor fingerprint is stable between different runs of the \acronym, data  was collected for $3$ hours for each sensor and then remove those from the tank and Arduino board (we refer to it as  "first run"). Then, in the second run,  the sensor was placed  back  on the tank and connected  with the Arduino board.  Features from the collected data are compared across in two runs. A visual representation of this experiment is shown in Figure~\ref{freq_response}. Here the first run is on the left hand and second  on the right hand column. The top pane shows the time domain data plotted for a sample of two sensors. The time domain data is the same from first run to second run. However, shape and variance of time domain signal looks similar for the two sensors. To further distinguish and to motivate the need for using spectral features,  frequency response (Fourier Transform) was plotted in the middle and bottom pane for sensor1 and sensor2, respectively. From the  middle and bottom panes, we can see the frequency bins (x-axis) for peak magnitude is different for both the sensors even though the signal might look similar in time domain, but they have distinguishable spectral features. In Figure~\ref{freq_response}, if we look horizontally we can see that these features are present from one run to the second, thus showing the stability of sensor fingerprints.

 \begin{figure}[!htb]
\centering
\includegraphics[height=7cm,width=12cm,keepaspectratio]{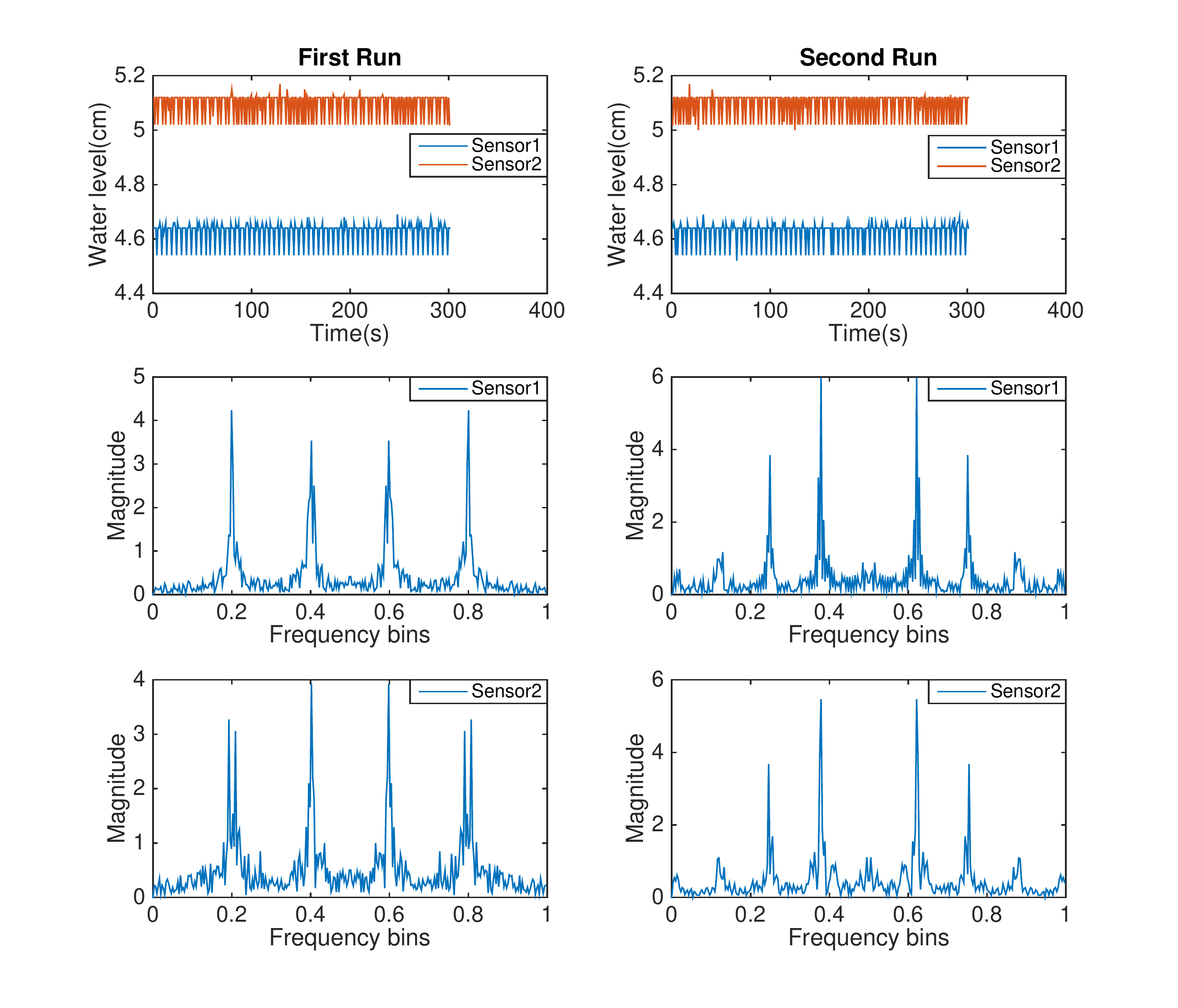}
\vspace{-2em}
\caption{Stability of the \acronym for the case of two small ultrasonic sensors. Run 1 corresponds to data collected for three hours, run 2 for two hours after removal and re-installation of sensor. In time domain both sensors time series look similar but frequency response is distinguishable and also stable for different runs.}
\label{freq_response}
\end{figure}

\begin{table}
\begin{center}
\caption{Randomly selected sensors vs. accuracy. 
}
\label{5to20sensors_accuracy}
 \begin{tabular}[!htb]{|c | c|} 
 \hline
 Number of Sensors  &  Accuracy (Chunk Size (s) = 120)  \\ 
 \hline
 5 & 96.97\% \\ 
 \hline
 10 & 97.92\%  \\
 \hline
 15 & 97.65\% \\
 \hline
 20 & 96.43\% \\ [1ex] 
 \hline

\end{tabular}
\end{center}
\end{table}

\begin{table}
\begin{center}
\caption{Analog Sensor Spoofing Attack Detection.}
\label{sensor_spoofing_accuracy}
 \begin{tabular}[!htb]{|c | c|} 
 \hline
 Number of Sensors  &  Attack Detection Accuracy   \\ 
 \hline
 10 & 100.0\% \\ 
 \hline
 7 & 99.64\%  \\
 \hline
 3 & 99.28\% \\ [1ex] 
 \hline

\end{tabular}
\end{center}
\end{table}

\subsection{Analog Sensor Spoofing Attack Detection}
RQ6: {\em Is the sensor fingerprint based method, able to detect analog sensor spoofing attack? } We extend attack detection performance of the proposed sensor fingerprinting method to detect analog sensor spoofing attacks. Experiments are designed, on attacks reported in the literature~\cite{shoukry2015}. Such an attacker model is shown in Figure~\ref{analog_sensor_spoofing} where an attacker spoofs the reading before it reaches the sensing device. Experimental setup is shown in Figure~\ref{sensor_spoofing_exp_photo}, in Appendix~\ref{supporting_figures_appendix}. The experiment was designed to measure the distance  from a wall using active ultrasonic sensors. An active sensor has a transmitter (TX) that transmits a probe signal, and a receiver (RX) that listens to a response signal based on which the physical quantity is measured. It can also be applied in a similar way for the case of water level measurement in a water treatment plant. A legitimate sensor is placed at a fixed distance from the wall and an attacker is brought to   transmit a response signal that  deceives the sensor into believing that the spoofed signal is the actual echo from the wall. Data was collected before and after the start of the attack. Results are shown in Table~\ref{sensor_spoofing_accuracy}. For $20$ ultrasonic sensors, the attack is detected soon after it is launched. Data was collected and labeled in attack and attack-free scenarios. SVM was used to classify the presence of an attacker based on ground truth (labeled data). The intuition behind sensor fingerprint is that it is a unique characteristic for a pair of transmitter (TX) and  receiver (RX)  for an active sensor. Since, during this attack, RX is actually getting a sound wave from another TX, violating that fingerprint. This observation is intuitive, as presence of an active attacker in victim's vicinity, raise the energy (impeding sound waves on transducer) in the environment~\cite{ghosttalk_2013} and changes the noise pattern for the sensor's receiving transducer~\cite{buchla1992applied}. 


\Paragraph{Theoretical Proof} In the following,  theoretical guarantee is provided for detection of analog sensor spoofing attacks.   
To understand this, we need  look at how these sensors work. For example, in ultrasonic sensors, sound waves vibrate the diaphragm of the transducer and sound energy is converted into electrical signal and appears at the output of the transducer. This produced voltage is proportional to the strength of sound vibrations and is the signal of interest. This analog signal has sensor noise effects in it. In the literature~\cite{buchla1992applied}, \emph{signal to noise ratio ($SNR$)} is used to analyze the effects of noise on a signal.  

\begin{definition}
Noise floor: The magnitude of noise in a sensing device is referred to as `noise floor'~\cite{buchla1992applied}. 
\end{definition}

\begin{definition}
Energy of a time domain signal ($\mathcal{S}$) can be calculated as:
\begin{equation}
\mathcal{S} = \int_{0}^{T} V^2 dt
\end{equation}
where $V$ is the voltage level of the signal, and $T$ is time window for which $\mathcal{S}$ is measured~\cite{ghosttalk_2013}.
\end{definition}

\begin{definition}
Signal to Noise Ratio (SNR): It is defined as a ratio of signal energy ($\mathcal{S}$) to noise energy ($\mathcal{N}$)~\cite{buchla1992applied}.
\end{definition}

\begin{thm}
Let $SNR_s = \frac{\mathcal{S}_s}{\mathcal{N}_s}$, be sensor's signal to noise ratio. Under analog sensor spoofing attack, noise floor changes and $SNR_s$ deviates from the one under normal operation, and the attack is detected.
\end{thm}

\begin{proof}
The proof of the theorem follows directly from \emph{Definitions~1-3} and a similar analog signal contamination argument made in~\cite{ghosttalk_2013}. Suppose the signal to noise ratio at a sensor is given as $SNR_s = \frac{\mathcal{S}_s}{\mathcal{N}_s}$, where $\mathcal{S}_s$ is signal strength and $\mathcal{N}_s$ is sensor noise at the output of the transducer of the sensor. During the attack signal to noise ratio becomes $SNR_a = \frac{\mathcal{S}_s + \mathcal{S}_a}{\mathcal{N}_s}$, where $SNR_a$ is signal to noise ratio under attack and $\mathcal{S}_a$ is the attacker's signal energy. Since $\mathcal{S}_s + \mathcal{S}_a \neq \mathcal{S}_s$, hence  $SNR_a \neq SNR_s$. From this result we can see that in the presence of an attacker, \emph{noise floor} is changed and since \acronym,  is based on sensor's noise pattern, attack will be detected as sensor's fingerprint could not be authenticated.   
\end{proof}

\subsection{Discussion} \label{discussion_sec}

\Paragraph{Accuracy}
The proposed method is intended to complement the existing intrusion detection systems. Legacy intrusion detection systems are either host based or network based~\cite{CPS_IDS_survey2014}, and hence those may not detect physical or analog sensor spoofing attacks. The method  proposed in this work is based on the physical structure and hardware of the devices, making it an intrusion detection system for physical attacks. However, because of the very nature of the proposed scheme, feature vector has some randomness and few sensors occasionally have overlapping feature sets. This leads to an identification accuracy less than $100\%$. Nevertheless, in the majority of the cases, sensor identification accuracy more than $97\%$ was achieved which makes our proposed method comparable to those proposed in the literature. 

\Paragraph{Scalability}
Results in the previous section indicate that \acronym scales well with an increase in the number of sensors to be fingerprinted. However, in comparison with the  device fingerprinting work in CPS domain, authors in~\cite{raheem2016}, created fingerprints for $2$ relay devices, from two different vendors, with an accuracy of $92\%$. In comparison, our method performs better with an accuracy of $97\%$ for tens of different sensors of same model.

\Paragraph{Robustness against Forgery}
Even though an adversary had learned the sensor noise, we believe each new device has it's own noise, so when an adversary replaces a device, it can not modify it's hardware to get the same noise pattern as another device without affecting its  performance or intended application.  Physically tampering a the device would change the noise pattern. If we assume a powerful attacker who can learn the noise distribution and reproduce it, our proposed scheme is still helpful for raising the bar for attackers. Learning hardware characteristics, for all the hardware devices in the network, is time consuming  and will likely raise suspicion and ultimately reveal the presence of an attacker to system operators.

\Paragraph{Application in Real-World CPS}
We have collected and tested the proposed method for a data set collected over a period of six months in a real water treatment testbed and for a couple of week data from a water distribution testbed. The results are promising for such a time period. However, it is recommended to train the classifiers after every plant maintenance cycle. Moreover, being used in a testbed for six months is different from being used in a real world production system of physical plants with possibly more harsh environment especially for the case of level measurements including rivers, dams etc. Although the testbeds  used in the reported experiments imitates real water treatment plants as close as possible but we believe the sensors and actuators wear out with time, rendering them less accurate. There is a possibility that those environmental effects may change the fingerprint but according to our hypothesis each sensor will be affected in a distinct way and, if retrained, will  possess a unique fingerprint. As far as the ambient noise or interference is concerned that would affect all the devices in a similar manner, letting us to cancel out those effects from sensor fingerprint.  



\section{Towards Physical Quantity based \emph{Challenge-Response} Protocol} \label{challenge_response_protocol_sec}
The proposed \acronym works well in the case of physical manipulation or physical layer signal spoofing in the sensor. However,  it is challenged by attackers who can learn and inject the correct sensor noise in the digital domain while spoofing the real measurements. For example, in the case of a replay attack, sensor noise from previous readings also gets replayed while preserving the sensor noise fingerprint. Another example is a traditional \emph{man-in-the-middle} attacker who can learn the sensor noise pattern and, while injecting fake sensor measurements, also adds the sensor noise fingerprint making it hard for \acronym to detect such an attack. We propose a novel \emph{challenge-response} protocol in which a challenge is produced in the physical quantity to be measured. Traditionally, a challenge is generated in the digital domain\,\cite{bruno_mo_physicalwatermarking_2015,shoukry2015} and its effect is observed on sensor measurements.
 \begin{figure}[!htb]
\centering
\includegraphics[scale=.35]{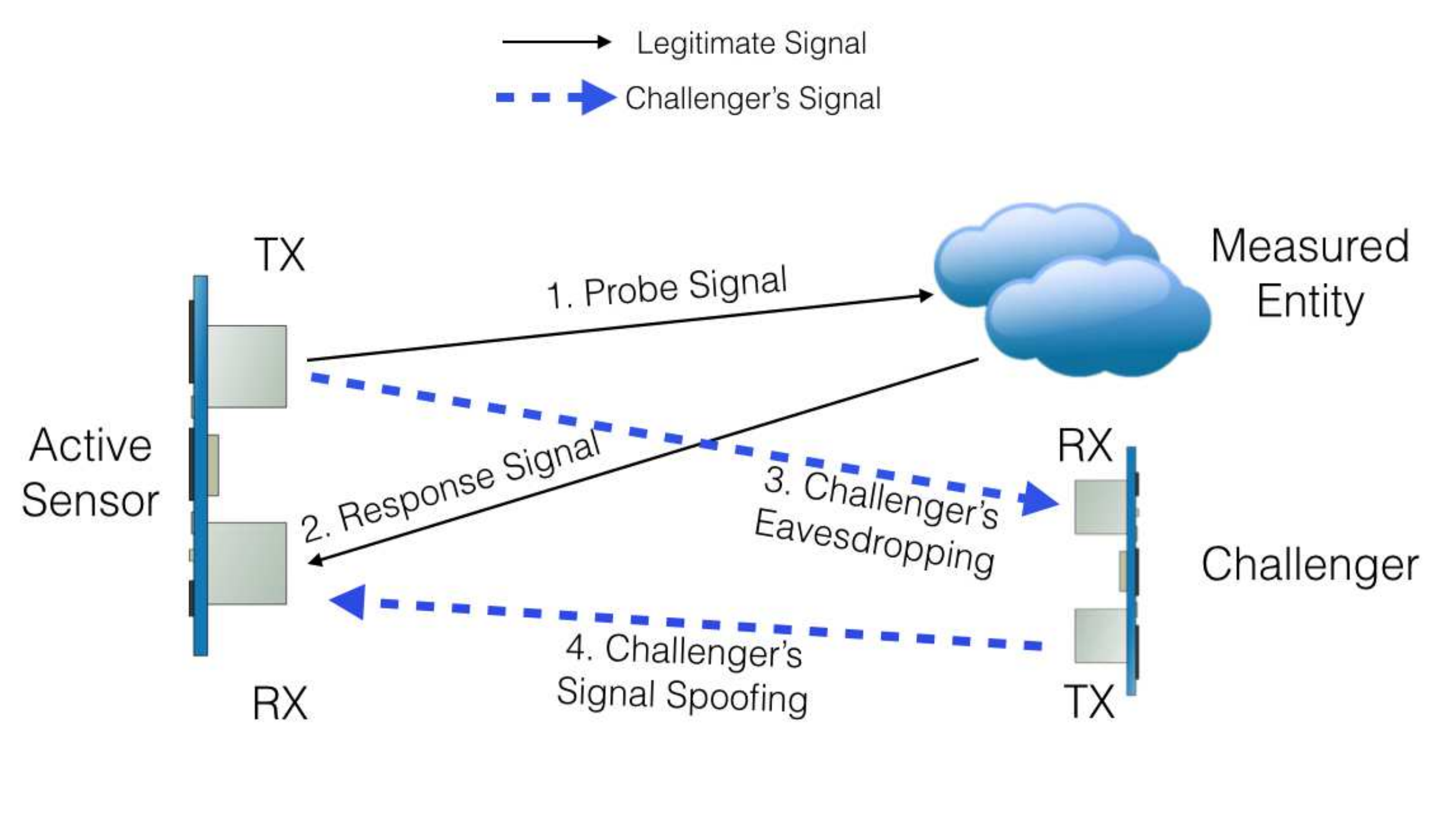}
\caption{Challenge-Response Protocol.}
\label{challenge_response_protocol}
\end{figure}
In our scheme, the challenge originates from the physical/analog domain. The challenger, will be between the physical quantity and the receiver~(RX) of the sensor. We send the challenge as a fake measurement and observe it's effect in sensor's output. This proposed \emph{challenge-response} protocol, enables \acronym, to detect a strong cyber adversary. 
Essentially, a special sort of analog spoofing attack is launched which, as discussed in the previous sections, is  
expected to be reflected and identifiable by the sensor readings, as this attack will change the noise fingerprint.


Consider the  attack scenario \emph{Replay and Advanced Sensor Spoofing Attack}( $\mathcal{S}8$) as explained in section~\ref{threat_model_sec}. The proposed \emph{challenge-response} protocol would detect replay attack as the attacker would not be able to replicate the effect of a fresh challenge. Now,  consider an even stronger attacker who can actively receive sensor data and tries to discover the presence of the \emph{challenge}. We train the machine learning algorithm on \emph{legitimate sensor-challenger pair fingerprint} in addition to training it for sensor noise fingerprint. With the \emph{challenge-response} protocol, \acronym tests sensor data for the presence of the sensor's as well as the challenger's fingerprint. In the presence of a challenge, an active attacker can observe the change in the sensor fingerprint but it needs time to learn the challenger's fingerprint. Machine learning based identification works on a chunk of data, and it takes time to collect that data and to make a decision. Therefore, an attacker needs some time and a chunk of data to learn challenger's fingerprint. After learning th challenger's fingerprint, an attacker might try to add the effect of challenger's noise in sensor measurements but it will appear with a delay (due to the time it took to learn the change), which will likely expose the attacker.  There is an important consideration in designing such a \emph{challenge-response} protocol; so that it should not effect the normal functionality of the industrial control system. 
 
\Paragraph{Design of Challenge-Response Protocol} In Figure~\ref{challenge_response_protocol}, a proposed setup for \emph{challenge-response} protocol is shown. A \emph{challenger} is a device similar to the one  used for sensing. This \emph{challenger} is placed between the active sensor and the measured entity. While active sensor waits for response to the probe signal, the \emph{challenger} transmits a signal to spoof the fake reading. This is similar to analog sensor spoofing attack as explained in section~\ref{threat_model_sec}. Unlike analog spoofing attack, the \emph{challenger} needs to spoof a reading hat is close to the reading that would be received by a real sensor. A naive \emph{challenger} would send a random reading that is not  close to the real quantity and we would be able to observe it's effect on the sensor's data but it might disturb the physical process as our \emph{challenger} itself is attacking sensor measurements. One important consideration in the design of \emph{challenger} is to not  disturb the control logic because of the added \emph{challenge} in the sensor data. Receiver (RX) is used as the challenger to receive the probe signal from the active sensor's transmitter (TX). Based on the probe signal, the \emph{challenger} can calculate a quantity to be spoofed. For example, if the sensor is an ultrasonic sensor used to measure the distance, then by listening to sensor's probe signal, the \emph{challenger} is able to calculate the exact quantity and it can generate the new spoofed signal which is close to the real quantity. Thus the spoofed signal will not trigger any wrong actuation due to the added \emph{challenge} in the sensor but  will change the noise fingerprint of the sensor. Such a \emph{challenger} may not work in all cyber physical systems but it is practical in the use cases studied in this work. In the case of an ultrasonic level sensor installed in the  testbeds used in this work,  the \emph{challenger} can be placed, for example, on water surface on a buoy, and it would be able to calculate the challenge signal which is close to the real system state. Following are the steps to create a \emph{challenge} as shown in Figure~\ref{challenge_response_protocol}:
\begin{enumerate}
  \item Transmitter (TX) of the active sensor sends a probe signal towards the quantity to be measured.
  \item Response signal arrives at the sensor's receiver (RX) resulting in measurement of the physical quantity.
  \item \label{receiver} Receiver (RX) of the \emph{challenger} is passively listening to probe signals. This eavesdropping will help it to calculate the amount of signal to be spoofed so that it does not disturb the control logic.
  \item Transmitter (TX) of the \emph{challenger} spoofs a signal to the receiver (RX) of the active sensor. The amount of fake measurement to be transmitted is calculated in step\,\ref{receiver}. This will keep the measurement close to the real quantity but change the noise fingerprint of the sensor.
\end{enumerate}

\Paragraph{Security argument}: The proposed challenge is an instance of an analog spoofing attack, which has been shown to alter
the noise profile of the sensor under test in the previous subsection. Therefore, by challenging a sensor at a random 
time $t$ for a random $\delta$ period, we expect the values of the sensor (as for instance reflected in the ICS historian) to have the (anomalous) noise fingerprint of an analog spoofing attack. If not, we can suspect a cyber attacker to be spoofing the historian values. An attacker who is aware of this \emph{challenge-response}, but at the same time is spoofing the real values, needs thus to consistently reflect the anomalous noise profile starting at time $t$ and for $\delta$ seconds. However, if the challenge is close to the real physical measurement, an attacker needs to wait $\Delta$ seconds in order to recognize that the noise profile has changed, at the beginning and at the end of the challenge. Therefore, he can at most react consistently with the expectations of the challenger at time $t + \Delta$ and stop at $t+\delta+\Delta$. As is shown in the previous sections,  $\Delta$ needed to detect a change in the noise profile is significant (around 1~minute) and can be leveraged to detect incoherent responses to the challenge.

\section{Related work}
\label{related_work_sec}
\emph{Device Fingerprinting}:
The approach presented in this article is inspired by the idea of using sensor noise as a fingerprint for camera identification\,\cite{lukas2006}. In\,\cite{lukas2006}, images are taken by a camera and filtered to obtain noise components and averaged for all images. This resultant noise vector acts as a reference pattern for test images. An image is tested against reference patterns for all cameras being studied and matched with one having the highest correlation with image's noise vector. In\,\cite{nick-2009} prospects of sensor identification, based on \cite{lukas2006}, are studied. The authors analyzed the effects of varying number of images to get a reference pattern. It is shown that the noise fingerprinting method  as in\,\cite{lukas2006}, can be used to counteract injection attacks at the time of checks in biometric systems as well as to establish evidence in a criminal incidence.
In \cite{khanna-2007} a method for image authentication from flatbed desktop scanners is presented. This method is similar to that used in earlier works for identifying digital cameras using pattern noise of the imaging sensor. In\,\cite{prabhu-2011} seven techniques for extracting unique signatures from NAND flash devices are studied. These techniques are experimentally evaluated using thirteen different flash memory devices. The techniques can help identify and authenticate electronic devices that use  flash devices. The idea of fingerprinting a device remotely based on it's hardware is presented in\,\cite{kohno2005}. Small microscopic deviations in device's clock\,\cite{moon1999,paxson1998} are used as fingerprint for the particular device.
In\,\cite{raheem2015} inter arrival time of packets is analyzed to fingerprint devices on a small campus network. In \cite{boris2009}, 50 RFID smart cards from the same manufacturer and type are tested for fingerprints. 

\skipnoindent \emph{Smartphone Fingerprinting}:
In\,\cite{dey-2014} hardware imperfections during the sensor manufacturing process are exploited as a fingerprint. Accelerometers in smartphones are used to create fingerprints. Experimental results show that when the properties of these imperfections are extracted, the device, and ultimately a user, can be identified. In\,\cite{das-2014} a smartphone speaker is used as a fingerprint component. During fabrication, subtle imperfections arise in device speakers which induce anomalies in produced sounds. Experimental results show high accuracy of the proposed method for speakers from same as well as different vendors. Another work\,\cite{bojinov2014} used sensor fingerprinting to identify a mobile device. In\,\cite{bojinov2014} the speakerphone-microphone system and the accelerometer in a typical smartphone were used to identify the mobile device. The article listed types of sensors and the related imperfections which can be useful in identifying systems. \cite{das2015} analyze techniques to mitigate device fingerprinting either by calibrating the sensors to eliminate the signal anomalies, or by adding noise that obfuscates the anomalies.

\skipnoindent\emph{Software based Fingerprinting}: A technique is proposed \cite{franklin2006}, that identifies the wireless device driver running on an IEEE 802.11 compliant device by passive monitoring. 
In \cite{desmond2008} unique devices over a Wireless Local Area Network (WLAN) are fingerprinted in a passive manner through the timing analysis of 802.11 probe request frames. 
Nmap \cite{nmap}, uses active fingerprinting by sending specific requests to determine operating systems and server versions.
On the other hand p0f \cite{p0f}, is a tool which can determine operating system and browser version of a client by passively monitoring TCP and HTTP header fields. 
The version and configuration information of the web browser will be transmitted to websites upon request \cite{eckersley2010}, which can be used to fingerprint the devices on which these browsers are running. 
 In \cite{olejnik2012}, indirect history data, such as information about categories of visited websites can also be effective in fingerprinting users. In \cite{gulmezoglu2017}, 25 applications on cloud are fingerprinted based on cache access pattern. 
 
\skipnoindent\emph{RF Fingerprinting}: 
 RF fingerprints have been proposed using
a measured temporal link signature to uniquely identify the link between a transmitter and a receiver\,\cite{patwari2007,faria2006}. A similar approach based on received signal strength is shown for the case study of a smart water treatment plant \cite{jay2017}. Modulation domain of RF signals is also studied to generate unique device fingerprints \cite{brik2008}. Researchers have shown that signal waveforms can also be used for device fingerprinting amid manufacturing inconsistencies \cite{hall2005,remley2005}. Besides wireless fingerprinting research has shown that it is possible to fingerprint devices based on Ethernet NICs analog signals \cite{gerdes2006}. 

\skipnoindent\emph{CPS Device fingerprinting}: In\,\cite{raheem2016} authors focus on the idea of device fingerprinting in ICS. One approach in\,\cite{raheem2016} is based on traditional network traffic monitoring and observing message response times, while the second approach is based on physical operation times of a device. Analysis is carried out on $2$ latching relays based on their operation timings. Another related work presented a preliminary study, on the idea of sensor fingerprinting in~\cite{ahmed_QRS2017}, using $2$ sensors, based on correlation analysis with an accuracy of $86\%$. However, to the best of our knowledge, this article present the first attempt to present a rigorous analysis on a multitude of sensors specific to ICS. Our machine learning based sensor fingerprint successfully identifies sensors with an accuracy as high as $99\%$. This is first work to present a holistic scheme to detect physical and cyber attacks using physical quantity based \emph{challenge-response} protocol.

\section{Conclusions}
\label{conclusions_sec}
In this paper we have presented a sensor identification technique for ICS. The proposed method creates a fingerprint of each sensor based on sensor's noise statistics. The feature set used is extracted from sensor noise and  is leveraged for sensor identification. Machine learning is used to identify the installed sensors. Uniqueness in the fingerprint is the result of manufacturing imperfections even for the same type and model of the device.
The results in this article point to an effective  method for hardware identification without affecting the performance of the host system. The results reported show that it is possible to identify whether test data is generated by the  installed sensor or by a replaced sensor. Out of a $44$ total sensors, $31$ were identified with an accuracy more than $96\%$ and $13$ with an accuracy above $90\%$. Also, using true positive and false positive rate as a metric, we got a true positive rate of $100\%$ for most of the sensors and a false positive rate of $0 \%$. Analog sensor spoofing attacks are detected with an accuracy of $100\%$. Moreover, a \emph{challenge-response} protocol is proposed that is able to detect advanced cyber attackers by means of sensor fingerprinting.

In the future, we plan to extend the proposed technique to fingerprint actuators as well. Also we plan to extend the feature set considered for fingerprinting based on the physical working principles of sensors. 




\bibliographystyle{unsrt}
\bibliography{references.bib}

\appendices
%
%

\section{Physical Sensor Swap Attack in SWaT Testbed}
\label{swap_attack_example}
\skipnoindent \textit{An Example Attack}: An attack scenario is created, comprising attacker goal of physical damage$(\mathcal{G}1)$,  an insider with the physical access $(\mathcal{C}1)$ and with an objective of deceiving the controllers $(\mathcal{O}3)$. \emph{Sensor swap attack}, as explained in Section~\ref{threat_model_sec}, can achieve attacker's goal as shown in following attack, which is launched on a real water treatment testbed. During the normal operation of the plant, if $s_{11}$ indicates a low water level in tank~1, $PLC_{1}$ issues a command to stop the outlet pump, to avoid dry running of the pump (which can result in physical damage). Similarly, if the tank~1 is full (based on $s_{11}$ values), the inlet to tank~1 should be closed to avoid flooding. But if we swap sensor $s_{11}$ with sensor $s_{22}$, the new sensor configuration is $s_{21}$ which means sensor~2 is placed on top of tank~1 (but this sensor is still transmitting data to $PLC_2$) and $s_{12}$ which represents sensor~1 on tank~2 (but this sensor is still transmitting data to $PLC_1$). Now, as soon as we start the attack, $s_{21}$ indicates that the tank level is high (which indeed is high as it is placed over tank~1 instead of tank~2), $PLC_{2}$ will keep running the  outlet pump ultimately resulting in pump running dry and damage (assuming that there is no mechanical interlock that prevents dry running). 


 \begin{figure}[t]
 \centering
\includegraphics[scale=.3]{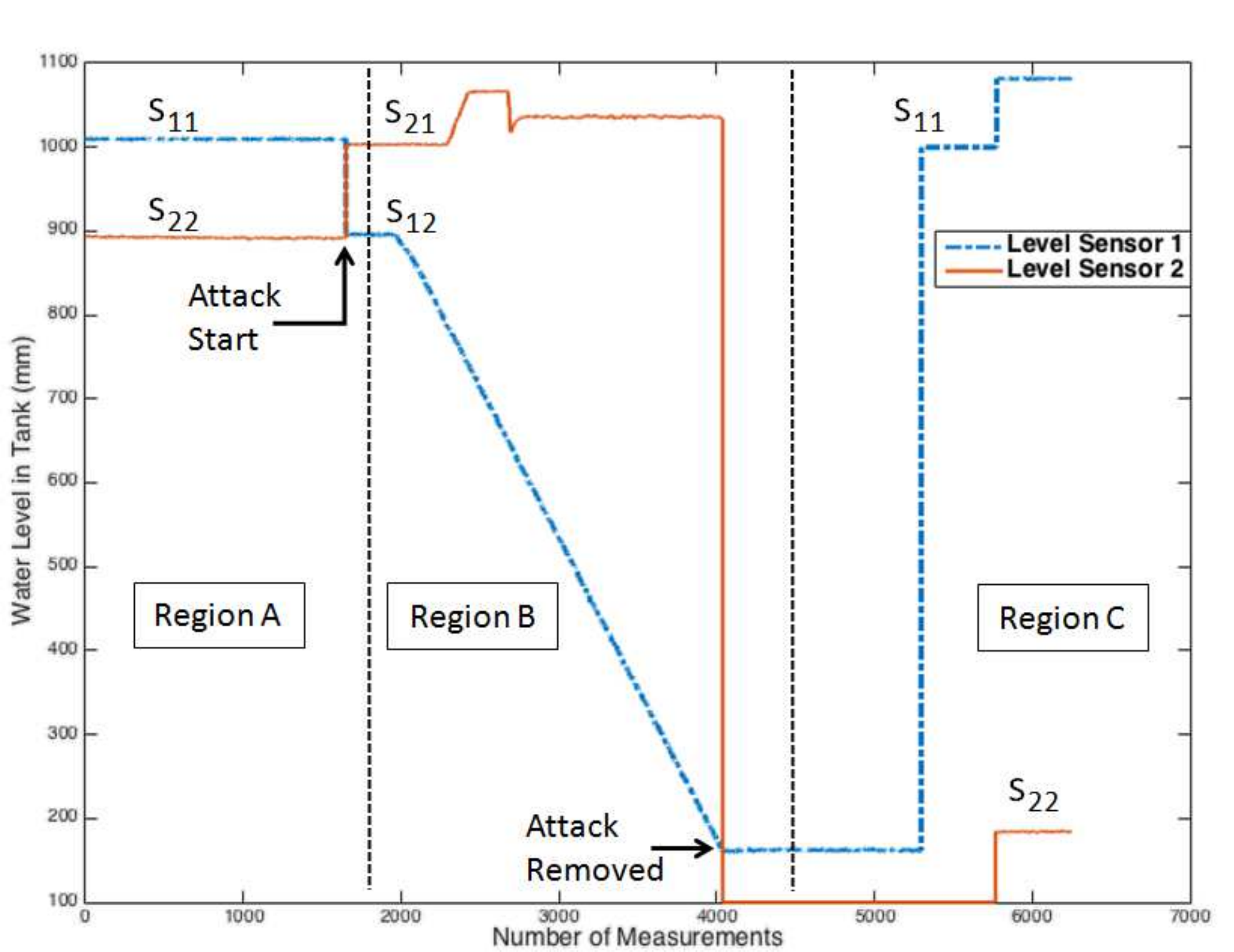}
\caption{Impact of swapping two sensors on the measurements received by corresponding PLCs. $s_{ij}$ represents $i^{th}$ sensor in $j^{th}$ tank. At the point of attack start $PLC_{1}$ will start receiving values from sensor 1 in tank 2 ($s_{12}$) and $PLC_{2}$ will receive $s_{21}$.}
\label{swap-attack}
\end{figure}  

\skipnoindent  Figure~\ref{swap-attack} illustrates the impact of a \emph{sensor swap attack}.  The attack start and  end times are  as labeled in Figure~\ref{swap-attack}.  When the attack starts, the level readings from $s_{11}$ and $s_{22}$ get exchanged. As the sensors are physically swapped, they keep sending data to their respective PLCs but those measurements are no longer coming  from the same normal process. If the  attack is not removed, the pump at stage two would be damaged, as it would continue to run without any water in the tank. The pump will stay ON because the corresponding level sensor is transmitting high water level to PLC. Note that as programmed in water treatment testbed, the PLC logic has   a low water level limit below which it is dangerous to keep pumping. However,   due to this attack, the PLC can be deceived. Regions A, B and C in Figure~\ref{swap-attack} show, respectively,   the sensor readings before, during, and after the attack. A detailed analysis on such an attack and preliminary results based on sensor fingerprint, is presented in form of a short paper~\cite{ahmed_QRS2017}.

\section{Support Vector Machine Classifier} 
\label{sec:svm}
SVM is a data classification technique used in many areas
such as speech recognition, image recognition and so on \cite{akata2014}. The aim of SVM is to produce a model based on the training data and give classification results for testing data. For a training set of instance-label pairs $(x_i, y_i), i = 1,...,k$ where $x_i \in  \mathbb{R} ^n$ and $y \in \{1, -1\}^k$, SVM require the solution of the
following optimization problem:

\begin{equation}
\begin{aligned}
& \underset{w,b,\zeta}{\text{minimize}}
& & \frac{1}{2}w^Tw + C \sum_{i=1}^{k} \zeta_i \\
& \text{subject to}
& & y_i(w^T\phi(x_i) + b) \geq 1 - \zeta_i, \\
&\text{where} \ \zeta_i \geq 0.
\end{aligned}
\end{equation}

The function $\zeta$ maps the training vectors into a higher dimensional space. In this higher dimensional space a linear separating hyperplane is found by SVM, where $C > 0$ is the penalty parameter of the error term. For the kernel function in this work we use the radial basis function:

\begin{equation}
K(x_i,x_j) = exp (-\gamma || x_i - x_j ||^2), \gamma > 0.
\end{equation}

In our work, we have multiple sensors to classify. Therefore, multi-class SVM library LIBSVM \cite{libsvm} is used.

\section{Water Treatment Testbed}
\label{swat-testbed}
It is a fully operational (research facility) scaled down water treatment plant producing 5\,gallons/minute of doubly filtered water, this testbed mimics large modern plants
for water treatment. Following is the brief overview of the testbed, for further details, please refer to~\cite{swat2016}.

\Paragraph{Water Treatment Process}
The treatment process consists of six distinct stages each controlled by an independent Programmable Logic Controller (PLC). Control actions are taken by the PLCs using data from sensors. Stage P1 controls the inflow of water to be treated by opening or closing a motorized valve MV-101. Water from the raw water
tank is pumped via a chemical dosing station (stage P2, chlorination) to another UF (Ultra Filtration) feed water tank in stage P3. A UF feed pump in P3 sends water via UF unit to RO (Reverse Osmosis) feed water tank in stage P4.
Here an RO feed pump sends water through an ultraviolet dechlorination unit controlled by a PLC in stage P4. This step is necessary to remove any free chlorine from the water prior to passing it through the reverse osmosis unit in stage P5. Sodium bisulphate (NaHSO3) can be added in stage P4
to control the ORP (Oxidation Reduction Potential). In stage P5, the dechlorinated water is passed through a 2-stage RO filtration unit. The filtered water from the RO unit is stored in the permeate tank and the reject in the UF backwash tank. Stage P6 controls the cleaning of the membranes in the UF
unit by turning on or off the UF backwash pump.
 
\Paragraph{Communication Network and Vulnerabilities}
Each PLC in the testbed obtains data from sensors associated with the corresponding stage, and controls pumps
and valves in its domain. PLCs communicate with each other through a separate network. Communications among sensors, actuators, and PLCs can be via either wired or wireless links.
Attacks that exploit vulnerabilities in the protocol used, and in the PLC firmware, are feasible and could compromise the communication links between sensors and PLCs, PLCs and
actuators, among the PLCs, and the PLCs themselves. Having compromised one or more links, an attacker could use one of several strategies to send fake state data to one or more PLCs.

\section{Water Distribution Testbed}
\label{wadi-testbed}
It is an operational testbed  supplying 10 US gallons/min of filtered water.  It represents a  scaled-down version of a large water distribution network in a city.  It contains three distinct control processes labeled P1 through P3, each controlled by its own set of PLCs. 

\Paragraph{Stages in WADI}
Water distribution process is  segmented into the following sub-processes:  P1: Primary grid, P2: Secondary grid, P3: Return water grid. 

\noindent{\em Primary grid}: The primary grid  contains two raw water tanks of 2500\,liters each, and a level sensor (1-LIT-001)  to monitor the  water level  in the tanks.  Water intake into these two tanks can be from the  water treatment plant,   from Public Utility Board inlet, or from the  return water grid.   A chemical dosing system is installed to maintain adequate water quality.  Sensors are installed to measure the water quality parameters of the water flowing into and out of the primary grid.

\noindent{\em Secondary grid}: This grid  has two elevated reservoir tanks and six consumer tanks. Raw water tanks supply water to the elevated reservoir tanks and, in turn,  these tanks supply water to the consumer tanks based on a pre-set demand pattern. Once consumer tanks meet their demands, water  drains to the return water grid. Return water grid is equipped with a tank. 

\Paragraph{Communications Infrastructure}

The communication network contains layer-0 (L0), layer-1 (L1) and layer 2-(L2). L0 is at process level and connects actuators/sensors and I/O modules via RS485-Modbus protocol. L1 is the plant control network where all PLCs are connected to a central node in a star topology. Communication among PLCs and RTUs takes place over Ethernet switches using NIP/SP based on TCP and High Speed Packet Access (HSPA) cellular gateways using GPRS modem. L2 is a communication network between a touch panel Human-Machine Interface (HMI) and the plant control network. This network is implemented using star topology and consists of PLCs and RTUs. A firewall  isolates the enterprise network from the plant control network. A SCADA workstation provides an interface between the plant operators  and PLCs  for remote monitoring and control.

\section{Supporting Figures/Tables}\label{supporting_figures_appendix}

 \begin{figure}
\centering
\includegraphics[scale=.3]{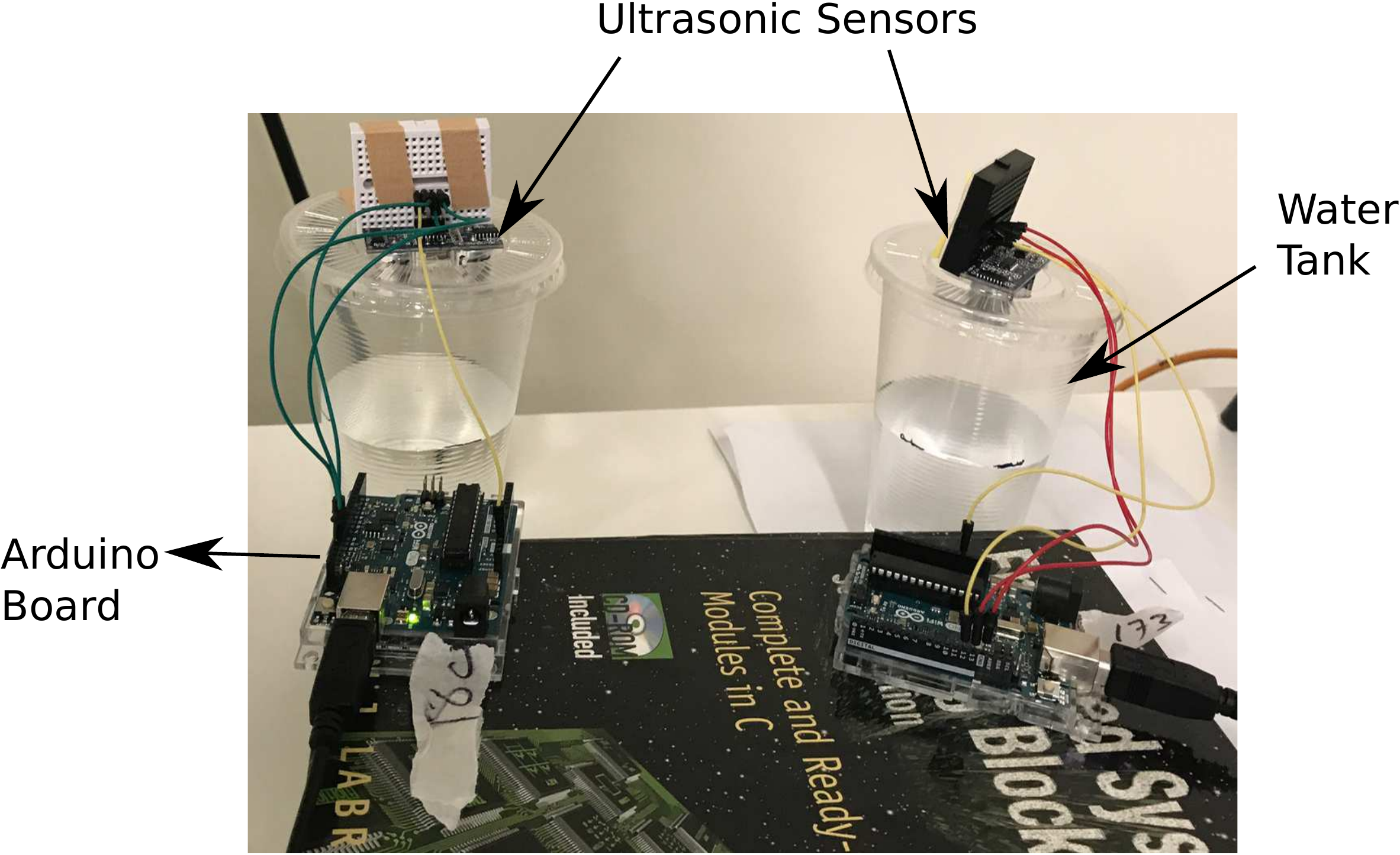}
\caption{Experiment setup with small ultrasonic sensors (HCSR04). A glass of water is filled up to a marked level to simulate the water tank. Sensors are controlled by Arduino Uno and data is collected.}
\label{hcsr04}
\end{figure}


 \begin{figure}[!htb]
\centering
\includegraphics[scale=.5]{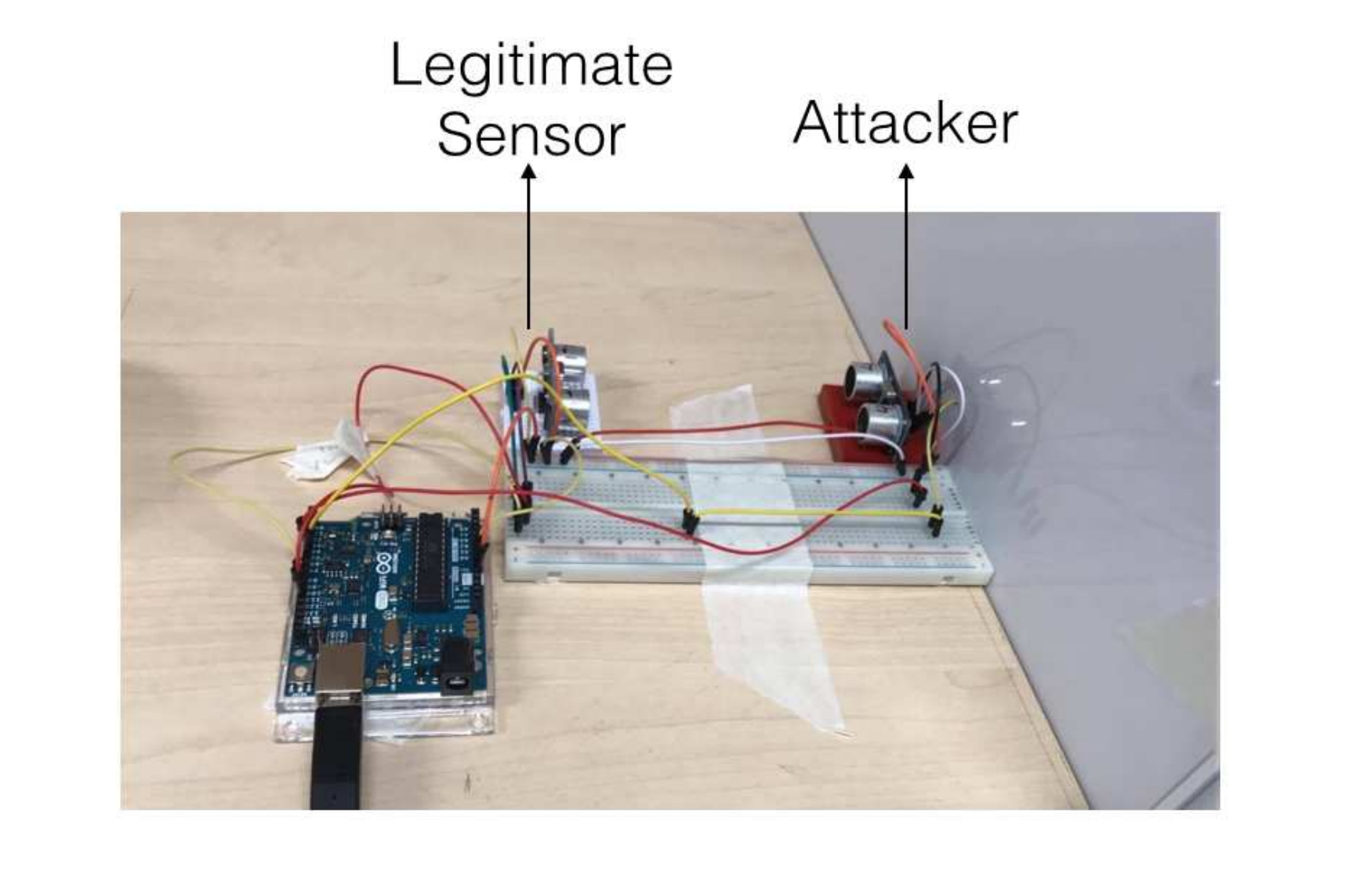}
\caption{Experiment setup with small ultrasonic sensors (HCSR04). A experiment is set to measure a fixed distance between the sensor and the wall. Sensors are controlled by Arduino Uno and data is collected.}
\label{sensor_spoofing_exp_photo}
\end{figure}


 \begin{figure}[!htb]
\centering
\includegraphics[scale=.45]{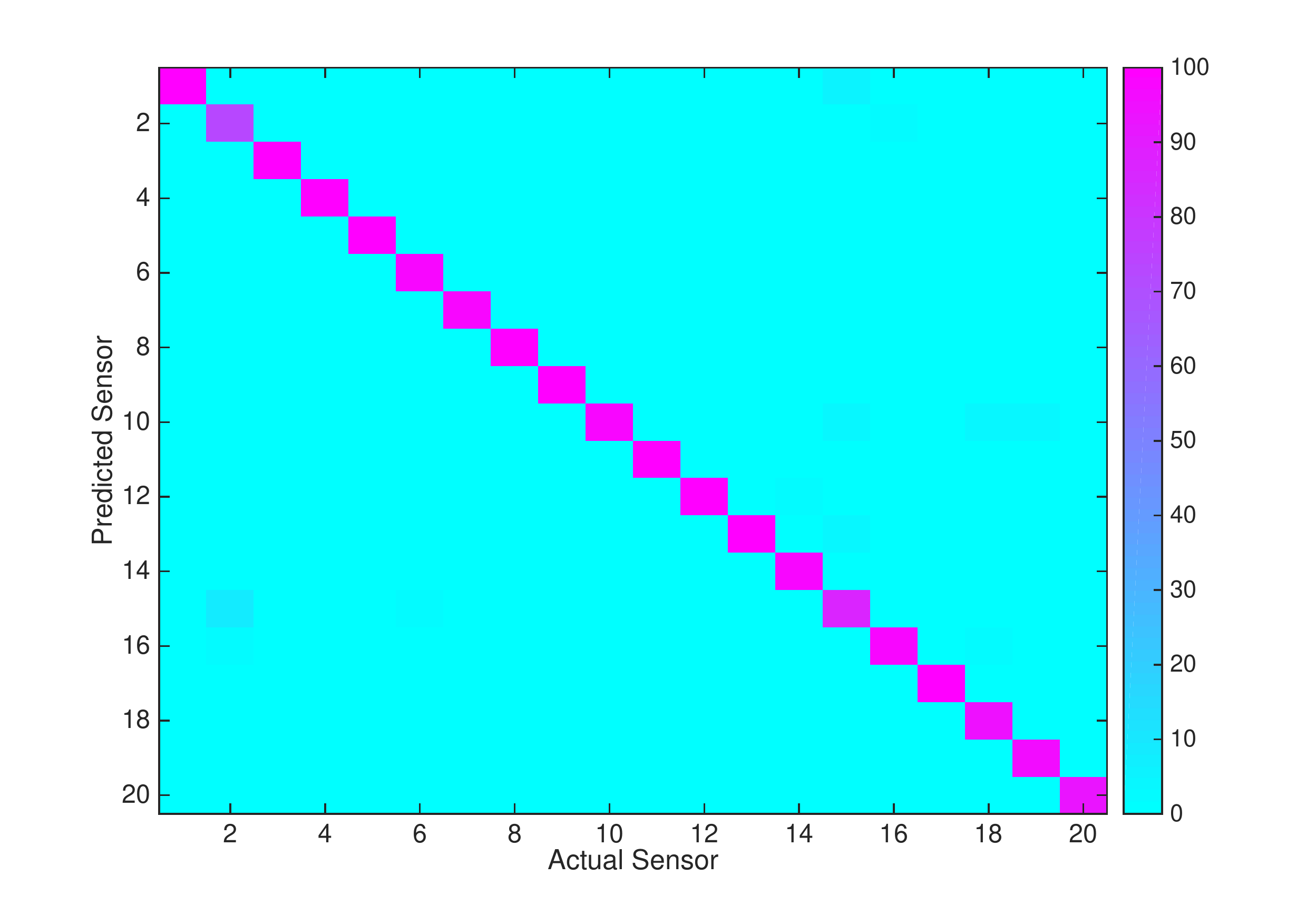}
\caption{Confusion matrix for 20 small ultrasonic sensors.}
\label{confusion_matrix}
\end{figure}

\end{document}